\documentclass{article}
\usepackage{graphicx}
\usepackage{mathtools}

\usepackage[toc,page,header]{appendix}
\usepackage{minitoc}
\usepackage[utf8]{inputenc}

\usepackage[british]{babel}
\usepackage{amsmath,amssymb,amsthm,amscd}

\usepackage{mathtools}
\usepackage{mathrsfs}
\usepackage[active]{srcltx}
\usepackage{pgf,tikz}
\usepackage{mathrsfs}
\usepackage{enumerate}
\usepackage{multirow}

\usepackage{braket}
\usetikzlibrary{arrows}
\usepackage{dsfont}
\usepackage{paralist}
\usepackage{comment}
\usepackage{enumerate}

\usepackage{soul}
\usepackage{fullpage}
\usepackage{bm}
\definecolor{navyblue}{rgb}{0.0, 0.0, 0.5}

\usepackage{hyperref}

\theoremstyle{plain}
\newtheorem{thm}{Theorem}[section]
\newtheorem{lem}[thm]{Lemma}
\newtheorem{prop}[thm]{Proposition}

\newtheorem{cor}[thm]{Corollary}

\theoremstyle{definition}

\theoremstyle{remark}

\newcommand*{\bid}{\mathbf{1}}

\newcommand*{\cC}{\mathcal{C}}

\newcommand*{\cD}{\mathcal{D}}
\newcommand*{\cE}{\mathcal{E}}

\newcommand*{\cG}{\mathcal{G}}
\newcommand*{\cH}{\mathcal{H}}
\newcommand*{\cI}{\mathcal{I}}

\newcommand*{\dI}{\mathbb{I}}

\newcommand*{\cN}{\mathcal{N}}

\newcommand*{\cL}{\mathcal{L}}

\newcommand*{\cP}{\mathcal{P}}

\newcommand*{\cS}{\mathcal{S}}
\newcommand*{\fS}{\mathfrak{S}}
\newcommand*{\cT}{\mathcal{T}}
\newcommand*{\cU}{\mathcal{U}}

\newcommand*{\cW}{\mathcal{W}}
\newcommand*{\cX}{\mathcal{X}}

\DeclareMathOperator{\var}{Var}

\newcommand*{\eps}{\varepsilon}

\newcommand*\diff{\mathop{}\!\mathrm{d}}

\newcommand*{\id}{\mathrm{id}}
\newcommand*{\tr}[1]{\mathrm{Tr}\left[#1\right]}

\newcommand*{\spr}[2]{\langle #1 | #2 \rangle}
\newcommand*{\proj}[1]{|#1\rangle\!\langle #1|}

\newcommand*{\mge}{\succcurlyeq}
\newcommand*{\mle}{\preccurlyeq}

\newcommand*{\exs}[2]{\mathbb{E}_{#1}\left[#2 \right]}

\usepackage{authblk}

\usepackage[maxnames = 99]{biblatex} 
\title{Error exponent of activated non-signaling assisted classical-quantum channel coding}

\author[1]{Aadil Oufkir}
\author[2]{Marco Tomamichel}
\author[1]{Mario Berta}
\affil[1]{\small{Institute for Quantum Information, RWTH Aachen University, Aachen, Germany.}}
\affil[2]{\small{Centre for Quantum Technologies \& Department of Electrical and Computer Engineering, National University of Singapore, Singapore.}}

\begin{document}

\maketitle

\begin{abstract}
We provide a tight asymptotic characterization of the error exponent for classical-quantum channel coding assisted by activated non-signaling correlations. Namely, we find that the optimal exponent\,---\,also called reliability function\,---\,is equal to the well-known sphere packing bound, which can be written as a single-letter formula optimized over Petz-R\'enyi divergences. Remarkably, there is no critical rate and as such our characterization remains tight for arbitrarily low rates below the capacity. On the achievability side, we further extend our results to fully quantum channels. Our proofs rely on semi-definite program duality and a dual representation of the Petz-R\'enyi divergences via Young inequalities.
\end{abstract}



\section{Introduction}

Shannon \cite{shannon1948mathematical} showed that the asymptotic transmission rate of classical information through a noisy channel can be characterized by a single number $C$, called the capacity. In fact, the capacity demarcates a transition point in the following sense: On the one hand, it is possible to send a number of messages $M$ with a rate $r=\frac{1}{n}\log M < C$ through $n$ copies of the noisy channel with an error probability decaying exponentially in $n$~\cite{shannon59,shannon67,fano61}. On the other hand, it is not possible to send information with a rate $r=\frac{1}{n}\log M > C$ above the capacity without incurring an error probability going to one exponentially in $n$~\cite{han89}. The notion of capacity was later generalized by Holevo \cite{Holevo1998Jan} to coding over channels with classical input and quantum output, so-called classical-quantum channels. Interestingly, the capacity of classical and classical-quantum channels remains the same when assistance via shared randomness, or quantum entanglement is available (see, e.g., \cite{Leung15}). These additional resources can be seen as particular cases of non-signaling assistance, which is assistance by devices shared between the sender and receiver that both parties can interact with, but that cannot by themselves be used for communication. In fact, the one-shot meta-converse of Polyanskiy, Poor, and Verdú \cite{polyanskiy2010channel} and Hayashi \cite{Hayashi09} exactly corresponds to the setting of non-signaling assisted strategies \cite{matthews2012linear}. In contrast, for classical-quantum channels, the meta-converse cannot be achieved with the help of plain non-signaling correlations, but rather requires the additional help of a single-bit of perfect communication~\cite{Wang2019Feb}.

Indeed, in the setting of classical-quantum channels, the meta-converse and the setting of non-signaling assistance provide two different semi-definite programming (SDP) relaxations for determining the optimal error probability for coding over classical-quantum channels. Understanding the relationship between these programs, and more generally, the impact of additional resources on higher-order refinements to Holevo's capacity theorem, can offer fundamental information-theoretic insights. First, converse expansions of the meta-converse immediately imply converses to plain, shared randomness, and entanglement-assisted codes. Second, following the rounding results of \cite{barman2017algorithmic,Fawzi19} for the success probability for coding, achievability expansions of the meta-converse can be translated into achievability bounds for plain codes (see also \cite{Vazquez-Vilar16}). This then gives deeper insights about the type of quantum R\'enyi divergences to find in formulas for higher order capacity refinements\,---\,which in turn allows to analyze the tightness of existing bounds on the error probability (including the necessity of critical rates and alike).

In the following, we focus on a particular higher-order refinement: the error exponent, also known as the reliability function. It governs the decay of the error probability of sending information with a fixed rate below the capacity over many identical copies of the channel. We study the classical-quantum error exponents for the case of non-signaling assisted codes as well as for the meta-converse. These two are exactly the same for classical channels and as such our work can be seen as a generalization of Polyanskiy's classical findings~\cite{Polyanskiy2012Dec}. 

We will next give an overview of our results. The precise definitions of the involved quantities are given in the subsequent section.


\paragraph{Overview of results:}

Given a classical-quantum channel $\cW=\{W_x\}_{x\in \cX}$ we are interested in the minimal average decoding error probability, $\eps(M, \cW)$, when sending one uniformly chosen message out of $M$ over the channel $W$.
Given a rate $r$ below the capacity $C(\cW)$, this optimal error probability decays like $ \eps^{}(e^{nr}, \cW^{\otimes n}) \approx  \exp(-n E(r))$, where 
\begin{align}
E(r):= \lim_{n\rightarrow \infty } -\frac{1}{n}\log  \eps^{}(e^{nr}, \cW^{\otimes n})
\end{align}
is called the error exponent or reliability function. As our main result, we determine the error exponent of the activated non-signaling error probability, $ \eps^{\rm{NS}}(e^{nr}, \cW^{\otimes n}\otimes \cI_2)$, for coding over the classical-quantum channel $\cW=\{W_x\}_{x\in \cX}$ with a rate $r\in(C_0(\cW),C(\cW))$ as
\begin{align}\label{EE-NS}
E^{\rm{ANS}}(r):= \lim_{n\rightarrow \infty } -\frac{1}{n}\log  \eps^{\rm{NS}}(e^{nr}, \cW^{\otimes n}\otimes \cI_2) = \sup_{\alpha \in(0,1]}  \frac{1-\alpha}{\alpha}\left( C_{\alpha}(\cW)-r\right),
\end{align}
where $\cI_2$ is a perfect one-bit channel permitting a so-called activation \cite{Wang2019Feb}, and the Petz R\'enyi capacities as
\begin{align}
C_{\alpha}(\cW):=\sup_{p\in \cP(\cX)} \inf_{\sigma\in \cS(\cH)} \exs{x\sim p}{D_{\alpha}(W_x \| \sigma )}.
\end{align}
with the Petz R\'enyi divergences~\cite{Petz86}, $D_\alpha(\rho\| \sigma):= \frac{1}{\alpha-1}\log \tr{\rho^{\alpha} \sigma^{1-\alpha}}$.

Our exact characterization \eqref{EE-NS} should be compared to the best known lower \cite{Dalai2013Sep} and upper bounds \cite{Renes2024Jul,Li2024Jul} on the plain error exponent
\begin{align}\label{EE-plain}
\sup_{\alpha \in [\tfrac{1}{2}, 1]} \frac{1-\alpha }{\alpha}\left( C_{\alpha}(\cW)-r\right)\le E(r) \le \sup_{\alpha \in(0,1]} \frac{1-\alpha }{\alpha}\left( C_{\alpha}(\cW)-r\right). 
\end{align}  
that only coincide when $r$ is above a critical rate\footnote{Note the change of parameter $\alpha = \frac{1}{1+s}$ so that $s=1 \Leftrightarrow \alpha =\frac{1}{2}$.}
\begin{align}\label{eq:critical-old}
r_c := \frac{\diff{\ }}{\diff{s}}\, s \, C_{\frac{1}{1+s}}(\cW)\, \Big|_{s=1}.
\end{align}
When $r<r_c$, the optimal plain error exponent is provably different from \eqref{EE-NS} and in general not even known for classical channels \cite{Csiszar2014Jul}. In contrast, when allowing non-signaling correlations with activation, we characterize the error exponent exactly for all rates and this bound corresponds to the upper bound on the plain error exponent in \eqref{EE-plain}. This then immediately gives an operational interpretation of the Petz R\'enyi divergence for $\alpha\in(0,1]$. For classical channels the same result was previously discovered by Polyanskiy for the meta-converse \cite[Theorem 24]{Polyanskiy2012Dec}, which exactly corresponds to non-signaling codes \cite{matthews2012linear}. Consequently, we then ask what happens for classical-quantum channels when using non-signalling assisted codes without activation. We show that whenever we have non-zero zero-error non-signalling assisted capacity or when the channel $\cW=\{W_x\}_{x\in \cX}$ exclusively features pure states $W_x$, then we still have the unconstrained characterization
\begin{align}
E^{\rm{NS}}(r)=E^{\rm{ANS}}(r).
\end{align}
Though for general classical-quantum channels, we can only show this for rates $r > \min\{r_c, r'_c\}$ with the (different) critical rate
\begin{align}
r'_c:=\sup_{\alpha\in (0,1)}(1-\alpha)^2 C_\alpha(\cW).
\end{align}
Lastly, we extend the achievability part of our main result \eqref{EE-NS} to fully quantum channels $\cW_{A\rightarrow B}$ and find
\begin{align}
E^{\rm{ANS}}(r)\ge\sup_{\alpha \in(0,1]}\sup_{\rho_R}\inf_{\sigma_B}\frac{1-\alpha}{\alpha} \Big(D_{\alpha}\big(\sqrt{\rho_R} (J_{\cW})_{RB} \sqrt{\rho_R}\big\| \rho_R\otimes  \sigma_B\big)-r\Big),
\end{align}
where $(J_{\cW})_{RB}=\sum_{i,j=1}^{|A|}\ket{i}\bra{j}_R\otimes \cW_{A\rightarrow B}\big(\ket{i}\bra{j}_A\big)$ denotes the Choi matrix of $\cW_{A\rightarrow B}$, $R\equiv A$, and the optimizations are over quantum states on systems $R$ and $B$, respectively. Whereas to the best of our knowledge no quantum sphere packing type lower bound on the error exponent are known, the best strongest entanglement-assisted upper bound takes the form \cite{Beigi2023Oct}
\begin{align}
E(r)\ge\sup_{\alpha \in [\tfrac{1}{2}, 1]} \sup_{\rho_R}\inf_{\sigma_B}\frac{1-\alpha}{\alpha} \Big(\widetilde{D}_{\alpha}\big(\sqrt{\rho_R} (J_{\cW})_{RB} \sqrt{\rho_R}\big\| \rho_R\otimes  \sigma_B\big)-r\Big).
\end{align}
It would be interesting to improve this to Petz R\'enyi divergences by extending the recent classical-quantum results from \cite{Renes2024Jul,Li2024Jul} to the fully quantum setting.

To conclude, and going back to the one-shot setting, the works \cite{barman2017algorithmic,Fawzi19} proposed a way to construct plain strategies from non-signaling ones via rounding strategies rooted in approximation algorithms. These rounding results\,---\,when applied for many identical copies of the channel\,---\,can be shown to be optimal for the first order capacity \cite{barman2017algorithmic,Fawzi19}, as well as the strong converse error exponent \cite{Oufkir24}. Motivated by this line of work (see also the related works \cite{Fawzi24,Ferme24} on network coding and \cite{Berta24,Cao24,Yao24} on channel simulation), establishing the exact error exponent for non-signaling strategies can be seen as a first step for understanding and potentially proving tight plain, shared randomness, or shared entanglement error exponents based on techniques from approximation algorithms. For example, does the entanglement-assisted error exponent for classical or classical-quantum channels feature a critical rate or not?\\

The rest of this manuscript is structured as follows. Before describing formally the non-signaling error probability in Section \ref{sec:NS coding}, we introduce some notation (Section \ref{sec:not-prem}) and recall the usual plain coding over classical-quantum channels (Section \ref{sec:plain}). In Section \ref{sec:NS-coding-CQ}, we then give our main achievability and converse results on the quantification of error exponents, which is followed by Section \ref{sec:improved rounding} on the discussion of activation and the need for critical rate(s).  
Finally, we present the extension to fully quantum channels in Section \ref{sec:NS-coding-QQ}, and conclude with some further open problems (Section \ref{sec:conclusion}).


\section{Preliminaries}

\subsection{Notation}\label{sec:not-prem}

The set of integers between one and an $n$ is denoted $[n]$. We consider an underlying Hilbert space $\cH$ of dimension $d$ and its tensor products $\cH^{\otimes n}$. The set of density operators on $\cH$ is denoted as $\cS(\cH)$. For two Hermitian operators $A$ and $B$, the notation $A\ll B$ means that the support of $A$ is included in the support of $B$. $A\mge B$ stands for $A-B$ is positive semi-definite.  
For a finite alphabet $\cX$, the set of probability distributions on $\cX$ is denoted as $\cP(\cX)$. 

The Petz R\'enyi divergences are defined for $\rho\ll \sigma$ and $\alpha\in [0,1) $ as \cite{Petz86}
\begin{align}
D_\alpha(\rho\| \sigma):= \frac{1}{\alpha-1}\log \tr{\rho^{\alpha} \sigma^{1-\alpha}}
\end{align}
and extended by continuity to the Umegaki relative entropy $ D(\rho\| \sigma) = \tr{ \rho (\log\rho -\log\sigma)}$ for $\alpha=1$. The sandwiched R\'enyi divergences between $\rho$ and $\sigma$ of order $\alpha$ are defined as \cite{Lennert13,Wilde14}
\begin{align}
\widetilde{D}_\alpha\left(\rho \middle\| \sigma \right) := \frac{1}{\alpha-1}\log \tr{\left(\sigma^{\frac{1-\alpha}{2\alpha}}\rho \sigma^{\frac{1-\alpha}{2\alpha}}\right)^{\alpha}}
\end{align}
if $\rho\sigma \neq 0$ and $\alpha \in (0,1)$ or $\rho \ll \sigma$ and $\alpha>1$. Otherwise, we set $\widetilde{D}_\alpha\left(\rho \middle\| \sigma \right) = +\infty$. The quantum hypothesis testing relative entropy is defined as \cite{Wang12}
\begin{align}
D_H^{\eps}(\rho\| \sigma) := -\log \min \Big\{\tr{\sigma O} \Big| \,0\mle O\mle \dI, \tr{\rho O}\ge 1-\eps \Big\}. 
\end{align}

For a Hermitian matrix $X$ with spectral decomposition $X = \sum_i \lambda_i \Pi_i $, we denote by $X_+ = \sum_{i} (\lambda_i)_+ \Pi_i $ and  $X_- = \sum_{i} (\lambda_i)_- \Pi_i $ its positive and negative parts, respectively. Let $|X| = X_+ +X_-$. The non-commutative minimum of two Hermitian matrices $A$ and $B$
is denoted $A\wedge B := \frac{1}{2}(A+B-|A-B|)$ and satisfies
\begin{align}
A\wedge B = A-(A-B)_+ = B-(B-A)_+. 
\end{align}
For $A,B \mge 0$, the non-commutative minimum satisfies \cite{helstrom1967detection,holevo1972analogue} 
\begin{align}\label{tr min A B min O}
\tr{A\wedge B} = \inf_{0\mle O\mle \dI} \tr{AO} +\tr{B(\dI-O)}. 
\end{align} 
as well as \cite{Audenaert2008Apr}
\begin{align}\label{Chernoff}
\tr{   A\wedge B } \le \inf_{0\le \alpha \le 1}\tr{A^\alpha B^{1-\alpha}}. 
\end{align}

A classical-quantum channel is a map from a finite alphabet $\cX$ to the set of density operators $\cS(\cH)$,  $\cW : x \mapsto W_x$ or equivalently a set of quantum state $\cW = \{W_x\}_{x\in \cX}\in \cS(\cH)^{\cX} $. For a probability distribution $p$ on $\cX$ we define the classical-quantum state
\begin{equation}
p\circ \cW := \sum_{x\in \cX} p(x) \proj{x} \otimes W_x,
\end{equation}
where $\{\ket{x}\}_{x\in \cX}$ is an orthonormal basis. We also abuse the notation and consider $p = \sum_{x\in \cX} p(x) \proj{x}$. The capacity of the classical-quantum channel $\cW$ is defined as
\begin{equation}
C(\cW) :=\sup_{p\in \cP(\cX)} D(p\circ \cW \| p\otimes p\cW ),
\end{equation}
where $p\cW = \sum_{x\in \cX} p(x) W_x$. More generally, the R\'enyi capacity of order $\alpha$ is defined as 
\begin{equation}
C_{\alpha}(\cW):=\sup_{p\in \cP(\cX)} \inf_{\sigma\in \cS(\cH)} D_{\alpha}(p\circ \cW \| p\otimes \sigma ) = \sup_{p\in \cP(\cX)} \inf_{\sigma\in \cS(\cH)} \exs{x\sim p}{D_{\alpha}(W_x \|  \sigma ) },
\end{equation}
where the last equality is proved in \cite{Mosonyi2017Oct}. Finally, the function  $\alpha \mapsto C_{\alpha}(\cW)$ is non-decreasing \cite[Proposition $2$ $(h)$]{Cheng2019Jan}. We refer to Appendix \ref{sec:Lemmas} for further technical lemmas from the literature needed for our proofs.


\subsection{Plain coding error probability}\label{sec:plain}

Given a classical-quantum channel $\cW$ that maps an input $x$ from a finite alphabet $\cX$ to a quantum state $W_x$, the coding problem is to transmit classical messages through an encoding, decoding scheme. More precisely, let $M$ be the number of the messages to be transmitted. An encoder is a mapping from the set of integers between $1$ and $M$ to elements in the alphabet $\cX$ or equivalently  a subset  $\cE=\{x(1), \dots, x(M)\}\in \cX^M$ constituted of elements from the alphabet $\cX$. A decoder is given by a POVM $\cD=\{\Lambda_{x(1)}, \dots, \Lambda_{x(M)}\}$ (a set of positive semi-definite matrices that sum to identity). A code is given by the encoding codebook and the decoding measurement $\cC= (\cE, \cD)$.  The error probability of transmitting the  $m^{\text{th}}$ message $x(m)$ is 
\begin{align}
\eps(x(m), \cW ) := 1-\tr{W_{x(m)}\Lambda_{x(m)}}
\end{align}
and the plain average error probability is then given by the expression 
\begin{align}
\eps(\cC, M, \cW) := \frac{1}{M}\sum_{m=1}^M\eps(x(m))= 1-\frac{1}{M}\sum_{m=1}^M\tr{W_{x(m)}\Lambda_{x(m)}}. 
\end{align}
The optimal plain error probability is given by optimizing over all possible codes $\cC = (\cE, \cD)$
\begin{align}\label{plain-error-one-shot}
\eps( M, \cW) :=\inf_{\cC} \eps(\cC, M, \cW).
\end{align}
In Section \ref{sec:NS coding}, we describe a possible relaxation of the optimization \eqref{plain-error-one-shot} obtained by allowing the sender and receiver to share non-signaling correlations.

With $n$ uses of the (iid) channel $\cW^{\otimes n}$, transmitting a number of messages with a rate $r$ corresponds to $M=e^{nr}$. A message $m\in [M]$ is encoded to a codeword $x_1^n(m)=  x_1(m) \cdots x_n(m)$ and sent through the iid channel $\cW^{\otimes n}$. The decoder receives the state $W^{\otimes n}_{x_1^n(m)} =  W_{x_1(m)}\otimes \cdots \otimes  W_{x_n(m)}$ and performs the measurement $\{\Lambda_{x_1^n(m')}\}_{x_1^n(m')}$.  
The average error probability takes the form
\begin{align}
\eps(\cC, M, \cW^{\otimes n}) = 1-\frac{1}{M}\sum_{m=1}^M\tr{W^{\otimes n}_{x_1^n(m)}\Lambda_{x_1^n(m)}}.
\end{align}

It is known that to send messages with a rate strictly smaller than the capacity of a channel the error probability can be exponentially small. More precisely, for a given  classical-quantum channel $\cW$ with capacity $C(\cW)$, the error probability of sending $M=e^{nr}$ messages (where $r<C(\cW)$) with $n$ copies of the channel, i.e $\cW^{\otimes n}$, has the form 
\begin{align}
\eps(e^{nr}, \cW^{\otimes n})  =   e^{-n E^{}(r) + o(n)}
\end{align}
for some function $r \mapsto E^{}(r)$ called error exponent or reliability function.


\section{Meta-converse and non-signaling coding error probability}
\label{sec:NS coding}

The starting point for our derivations is the meta-converse for the one-shot coding error probability \cite{polyanskiy2010channel,Matthews2014Sep}. Recall that for a classical-quantum channel $\cW = \{W_x\}_{x\in \cX} $ and a probability distribution $p\in \cP(\cX)$,  we define $p\circ \cW = \sum_{x\in \cX} p(x) \proj{x} \otimes W_x$ and write $p= \sum_{x\in \cX} p(x) \proj{x}$. Then, the meta-converse error probability is
\begin{align}
\eps_{}^{\rm{MC}}(M,\cW) &:= \inf_{} \Big\{ \eps : \sup_{p\in \cP(\cX)} \inf_{\sigma\in \cS(\cH)} D_H^{\eps}(p\circ \cW\| p\otimes \sigma )\ge \log M \Big\}
\\&= \inf_{p\in \cP(\cX)}\sup_{\sigma\in \cS(\cH)}\inf_{0\mle O\mle \dI} \Big\{\tr{p\circ \cW(\dI-O)} \;\Big|\; \tr{p\otimes \sigma O}\le \tfrac{1}{M} \Big\}. \label{MC}
\end{align}
Classically, it is known \cite{matthews2012linear,Wang2019Feb} that the meta-converse corresponds to the setting where the sender and the receiver share non-signaling (NS)  correlations. A NS-assisted code corresponds to a super-channel which is non-signaling from the encoder to the decoder and vice-versa. More precisely, a non-signaling channel between Alice and Bob whose input systems are $A_i, B_i$ and output systems are $A_o, B_o$ has a  positive semi-definite Choi matrix $J_{A_iA_oB_iB_o}$ satisfying (e.g., \cite{Fang2019Sep})
\begin{align}
J_{A_iB_iB_o}  &= \mathds{1}_{A_i} \otimes J_{B_iB_o}, &&(A\nrightarrow B), 
\\J_{A_iA_oB_i}  &= J_{A_iA_o}\otimes \mathds{1}_{B_i},  &&(B\nrightarrow A). 
\end{align}
Let $\cN$ be the set of non-signaling super-channels. The optimal non-signaling success probability of a classical quantum-channel $\cW = \{W_x\}_{x\in \cX}$ for sending $M$ messages is 
\begin{align}\label{eq:eps-NS-def}
1-\eps_{}^{\rm{NS}}(M,\cW):=\sup_{\Pi \in \cN}  \frac{1}{M} \sum_{m=1}^M \bra{m}\Pi (\cW) (\proj{m}) \ket{m},
\end{align}
where $\Pi ( \cW)$ is the effective channel. It is known that this success probability is the solution of the following semi-definite program (SDP) \cite{Matthews2014Sep,Wang2019Feb,Fawzi19}
\begin{equation}\label{NS-program}
\begin{split}
1-\eps_{}^{\rm{NS}}(M,\cW) =\sup_{\Lambda, p} &\quad\frac{1}{M}\sum_{x\in \cX} \tr{\Lambda_x W_x}  
\\\text{s.t.}& \quad  \sum_{x\in \cX} \Lambda_x = \mathbb{I}
\\&\quad 0\mle  \Lambda_x \mle p(x) \mathbb{I}  \qquad \forall {x\in \cX}       
\\&\quad  \sum_{x\in \cX} p(x)= M.
\end{split}
\end{equation}
As remarked by \cite{Wang2019Feb}, this program differ from the meta-converse program 
\eqref{MC} in the constraint $\sum_{x\in \cX} \Lambda_x = \mathbb{I} $ that should be $\sum_{x\in \cX} \Lambda_x \mle  \mathbb{I}$ in the latter (see \eqref{NS-program-le} later on). Note that we always have the inequality $ \eps_{}^{\rm{MC}}(M,\cW) \le \eps_{}^{\rm{NS}}(M,\cW)$. Moreover, \cite{Wang2019Feb} relates these two apparently different programs using the notion of activated capacity. More precisely, they show in \cite[Lemma 3]{Wang2019Feb} that allowing a perfect communication of a single bit using the channel
\begin{align}
\cI_2(\cdot):= \sum_{i=1}^2 \bra{i}(\cdot)\ket{i} \proj{i}
\end{align}
is sufficient to approximately reduce the non-signaling error probability to the meta-converse one
\begin{align}
\eps_{}^{\rm{MC}}(M,\cW)=\eps_{}^{\rm{MC}}(2M,\cW\otimes \cI_2) \le \eps_{}^{\rm{NS}}(2M,\cW\otimes \cI_2) \le  \eps_{}^{\rm{MC}}(M,\cW),
\end{align}
where the first equality follows from \cite[Eq $(11)$]{Wang2019Feb}. Hence, we have that 
\begin{align}\label{NS-MC-inequ}
\eps_{}^{\rm{MC}}(M,\cW)= \eps_{}^{\rm{NS}}(2M,\cW\otimes \cI_2). 
\end{align}
Note that for classical channels we have $\eps_{}^{\rm{MC}}(M,\cW) = \eps_{}^{\rm{NS}}(M,\cW)$ \cite{matthews2012linear}. This means that classically no activation is needed and the meta-converse  corresponds exactly to non-signaling coding.

The relation between non-signaling and the meta-converse in \eqref{NS-MC-inequ} motivates the study of error exponents of the activated non-signaling error probabilities defined as
\begin{align}
E^{\rm{ANS}}(r):= \lim_{n\rightarrow \infty } - \frac{1}{n}\log \eps^{\rm{NS}}(e^{nr}, \cW^{\otimes n}\otimes \cI_2). 
\end{align}
Using \eqref{NS-MC-inequ}, we deduce that the activated non-signaling and the meta-converse have the same error exponents
\begin{align}
E^{\rm{ANS}}(r)=E^{\rm{MC}}(r) := \lim_{n\rightarrow \infty } - \frac{1}{n}\log \eps^{\rm{MC}}(e^{nr}, \cW^{\otimes n}). 
\end{align}
Hence, it suffices to focus on the meta-converse for our derivations of the activated non-signaling error exponents. The meta-converse program \eqref{MC} can be written as the SDP \cite{Wang2019Feb}
\begin{equation}\label{NS-program-le}
\begin{split}
1-\eps_{}^{\rm{MC}}(M,\cW) =\sup_{\Lambda, p} &\quad\frac{1}{M}\sum_{x\in \cX} \tr{\Lambda_x W_x}  
\\\text{s.t.}& \quad  \sum_{x\in \cX} \Lambda_x \mle  \mathbb{I}  
\\&\quad 0\mle  \Lambda_x \mle p(x) \mathbb{I}  \qquad \forall {x\in \cX}        
\\&\quad  \sum_{x\in \cX} p(x)= M.
\end{split}
\end{equation}
This further relaxation has an advantage that can be seen in the dual formulation of the program  \eqref{NS-program-le} stated in the following proposition.

\begin{prop}\label{prop-dual}
Let  $\cW = \{W_x\}_{x\in \cX}$ be a classical-quantum channel. The optimal one-shot meta-converse and (activated) non-signaling error probabilities satisfy
\begin{align}
\eps_{}^{\rm{NS}}(2M,\cW\otimes \cI_2) = \eps_{}^{\rm{MC}}(M,\cW) &= \inf_{p \in \cP(\cX)} \sup_{B\mge 0} \sum_{x\in \cX} p(x) \tr{W_x\wedge MB}-\tr{B},\label{ANS-dual}
\\\eps_{}^{\rm{NS}}(M,\cW) &= \inf_{p \in \cP(\cX)} \sup_{B=B^\dagger} \sum_{x\in \cX} p(x) \tr{W_x\wedge MB}-\tr{B}.\label{NS-dual}
\end{align}
\end{prop}

The classical version of this result appeared in \cite[Proposition 14]{matthews2012linear}. The advantage of the formulation \eqref{ANS-dual} is that it relates the error probability to a (symmetric) hypothesis problem with a correction term (see \cite{Hayashi2007Dec} for a similar remark). Furthermore, the optimization in \eqref{ANS-dual} is over positive semi-definite matrices $B\mge 0$ which is more convenient for proving the achievability result than an optimization over Hermitian matrices $B=B^\dagger$ one would have in the  dual program \eqref{NS-dual} of the (non-activated) non-signaling error probability \eqref{NS-program}.

\begin{proof}[Proof of Proposition \ref{prop-dual}]
The proof is similar to \cite{matthews2012linear} and normalizing $p$ in \eqref{NS-program-le}, we get
\begin{equation}
1-\eps_{}^{\rm{MC}}(M,\cW) =\sup_{p \in \cP(\cX)} \sup_{\Lambda} \Big\{\tfrac{1}{M}\textstyle\sum_{x\in \cX} \tr{\Lambda_x W_x} \Big| \textstyle\sum_{x\in \cX} \Lambda_x \mle  \mathbb{I}  , \; 0\mle  \Lambda_x \mle M p(x) \mathbb{I}  \; \forall {x\in \cX}  \Big\} .
\end{equation}
The Lagrangian of the second maximization program is
\begin{align}
\cL(\Lambda,  B, C) &=  \frac{1}{M}\sum_{x\in \cX} \tr{\Lambda_x W_x}  +  \mathrm{Tr}\Big[ B\Big(\dI- \sum_{x\in \cX} \Lambda_x\Big) \Big] +\sum_{x\in \cX} \tr{C_x (M p(x)\dI-\Lambda_x) }
\\&= \sum_{x\in \cX} \mathrm{Tr}\Big[\Lambda_x \Big(\frac{1}{M}W_x -B-C_x\Big)\Big]+\sum_{x\in \cX} Mp(x)\tr{C_x} +\tr{B},
\end{align}
where $\Lambda_x, B, C_x\mge 0$.
Hence for all $x\in \cX$, we should have $C_x\mge \frac{1}{M}W_x -B $  so by the strong duality (the primal is bounded and strictly feasible) 
\begin{align}
\eps_{}^{\rm{MC}}(M,\cW) 
&= 1-  \sup_{p\in \cP(\cX)}\inf_{B,C\mge 0}  \sum_{x\in \cX} Mp(x)\tr{C_x} +\tr{B}
\\ &=\inf_{p\in \cP(\cX)} 1-  \inf_{B\mge 0} \sum_{x\in \cX} p(x)\tr{(W_x-MB)_+} +\tr{B}
\\&=\inf_{p\in \cP(\cX)} \sup_{B\mge 0}\sum_{x\in \cX} p(x) \tr{W_x- \left(W_x-MB\right)_+}- \tr{B}
\\&=\inf_{p\in \cP(\cX)} \sup_{B\mge 0}\sum_{x\in \cX} p(x) \tr{W_x\wedge MB}- \tr{B}.
\end{align}
The proof of \eqref{NS-dual} is the same except that we have the constraint $B=B^\dagger $ instead of $B\mge 0$ because $\sum_{x\in \cX}\Lambda_x = \dI$ in the non-signaling program \eqref{NS-program}.  
\end{proof}


\section{Non-signaling assisted error exponents}
\label{sec:NS-coding-CQ}

\subsection{Main result}

In Section \ref{sec:plain}, we recalled the plain coding over classical quantum channels. In particular, we defined the plain error exponent as the decaying rate of the optimal error probability for  sending information with rate below the capacity. Since non-signaling assisted codes are at least as powerful as their plain counterpart, the error probability should also decay exponentially with the number of channel uses below capacity
\begin{align}
\eps^{\rm{NS}}(e^{nr}, \cW^{\otimes n}) =  \exp(-n E^{\rm{NS}}(r)+o(n)).
\end{align}
Motivated by the meta-converse, we consider the activated non-signaling error probability as defined in Section \ref{sec:NS coding}. The goal of this section is to asymptotically characterize the corresponding error exponent $E^{\rm{ANS}}$ for activated non-signaling strategies and our main result is the following.

\begin{thm}\label{prop-NS-exponent}
Let $\cW$ be a classical-quantum channel and $r\geq0$. Then, we have that
\begin{align}\label{NS-exponent}
E^{\rm{ANS}}(r) = \sup_{\alpha \in(0,1]}  \frac{1-\alpha}{\alpha}\left( C_{\alpha}(\cW)-r\right).
\end{align}
\end{thm}

For rates below the zero R\'enyi capacity $r<C_0(\cW)$, the right hand side of \eqref{NS-exponent} diverges \cite[Proposition 4]{Cheng2019Jan} and the error probability becomes $\eps^{\rm{NS}} = 0$ (the optimal $\alpha \rightarrow 0$).\footnote{The zero R\'enyi capacity $C_0(\cW)$ is equal to the zero-error non-signalling assisted capacity \cite{Duan2015Dec}, as well as the zero-error capacity with feedback \cite{Dalai2013Sep}, and consequently gives an upper bound on the plain zero-error capacity.} For rates above the capacity $r > C(\cW)$, the error probability goes to $1$ and the right hand side of \eqref{NS-exponent} vanishes (the optimal $\alpha = 1$). So, the only non-trivial interval for the rates is $r\in(C_0(\cW), C(\cW))$. In the remainder of this note, we will focus on rates within this interval.

The expression in \eqref{NS-exponent} is exactly equal to the sphere-packing bound (SPB) and we stress that the existing proofs \cite{Dalai2013Sep,Cheng2019Jan} consider plain coding strategies. However, they can be adapted to handle the meta-converse and activated non-signaling assisted codes as well. In the following, we provide a slightly different proof of the SPB using the dual perspective (cf.~\ref{ANS-dual}). The achievability result in Theorem \ref{prop-NS-exponent} is stronger than the random coding achievability \cite{Beigi2023Oct,Renes2024Jul,Li2024Jul}, which might be expected since non-signaling strategies are at least as powerful as the plain, shared randomness and entanglement-assisted ones. We start by proving the achievability bound in Section \ref{sec-NS-exp-ach}, then we move to prove the converse bound in Section \ref{sec-NS-exp-cvs}.


\subsection{Achievability}
\label{sec-NS-exp-ach}

In this section, we prove the achievability result of Theorem \ref{prop-NS-exponent}. For this, we show a non-asymptotic upper bound on the non-signalling error probability.

\begin{prop}\label{prop-exp-ach}
Let  $\cW$ be a classical-quantum channel, $n\in \mathbb{N}$, $r\ge 0$, and $\alpha\in(0,1]$. Then, we have that
\begin{align}
\frac{1}{n}\log\eps^{\rm{NS}}(e^{nr}, \cW^{\otimes n}\otimes \cI_2) &\le   - \sup_{p\in \cP(\cX)}\inf_{\sigma\in \cS(\cH)}    \frac{1-\alpha}{\alpha} \Big(D_{\alpha}(p \circ \cW\| p\otimes  \sigma)-r\Big).
\end{align}
\end{prop}

Since it is known that \cite{Mosonyi2017Oct} 
\begin{align}
\sup_{p\in \cP(\cX)}\inf_{\sigma\in \cS(\cH)}  D_{\alpha}(p\circ \cW\| p\otimes  \sigma) = \sup_{p\in \cP(\cX)} \inf_{\sigma\in \cS(\cH)}  \exs{x\sim p}{D_{\alpha}( W_x\| \sigma )}
\end{align}
we immediately conclude that
\begin{align}
E^{\rm{ANS}}(r)\ge  \sup_{\alpha \in(0,1]} \sup_{p\in \cP(\cX)} \inf_{ \sigma\in \cS(\cH)} \frac{1-\alpha}{\alpha}\Big( \exs{x\sim p}{D_{\alpha}( W_x\| \sigma )}-r\Big) ,
\end{align}
which corresponds to the achievability part of Theorem~\ref{prop-NS-exponent}.

We note that an important step in the proof of Proposition \ref{prop-exp-ach} in order to arrive at the Petz R\'enyi divergences\,---\,instead of the sandwiched R\'enyi divergences as in \cite{Beigi2023Oct}\,---\,is to employ a Young type inequality (Lemma \ref{lem-young}) that becomes available in the dual picture because of the activated non-signaling assistance.

\begin{proof}[Proof of Proposition \ref{prop-exp-ach}]
Using Proposition \ref{prop-dual} we have that
\begin{align}
\eps_{}^{\rm{NS}}(M, \cW\otimes \cI_2)  &\le \eps_{}^{\rm{NS}}(2M, \cW\otimes \cI_2)=  \inf_{p\in \cP(\cX)}\sup_{B\mge 0} \sum_{x\in \cX} p(x) \tr{ W_x \wedge  MB }-\tr{B}. \label{dual-swap}
\end{align}
Since we have $ \tr{  W_x \wedge  MB }=\tr{ (\proj{x}\otimes W_x) \wedge (\proj{x}\otimes MB) } $  for all $x\in \cX$ and $(A,B)\mapsto \tr{A\wedge B}$ is concave (Lemma \ref{lem:concavity of min}) we deduce that
\begin{align}
\eps^{\rm{NS}}(M,\cW\otimes \cI_2) &\le  \inf_{p\in \cP(\cX)}\sup_{B\mge 0}  \sum_{x \in \cX} p(x) \tr{ (\proj{x}\otimes W_x) \wedge (\proj{x}\otimes MB) }-\tr{B} 
\\& \le  \inf_{p\in \cP(\cX)} \sup_{B\mge 0} \tr{\left(\sum_{x\in \cX} p(x)\proj{x}\otimes W_x\right) \wedge \left( \sum_{x\in \cX} p(x) \proj{x}\otimes MB\right)}-\tr{B}. 
\end{align}
Recall the notation $p \circ \cW = \sum_{x \in \cX} p(x)\proj{x}\otimes W_x$ and $p = \sum_{x \in \cX} p(x)\proj{x}$. Writing $B\mge 0$ as  $B= s\sigma$, where $s=\tr{B}$ and $\sigma\in \cS(\cH)$
 and using the inequality \eqref{Chernoff} we obtain for any $\alpha \in (0,1]$ that
\begin{align}
\eps^{\rm{NS}}(M, \cW\otimes \cI_2) &\le    \inf_{p\in \cP(\cX)} \sup_{\sigma \in \cS(\cH),\; s\ge 0}  \tr{\left(p \circ \cW \right) \wedge \left( p \otimes(Ms\sigma) \right)}-s
\\&   \le \inf_{p\in \cP(\cX)} \sup_{\sigma \in \cS(\cH),\; s\ge 0}  s^{1-\alpha}\tr{\left(p \circ \cW\right)^{\alpha}\left( p\otimes(M\sigma) \right)^{1-\alpha}}-s
\\&=  \inf_{p\in \cP(\cX)} \sup_{\sigma \in \cS(\cH),\; s\ge 0}  s^{1-\alpha} M^{1-\alpha} \exp\left((\alpha-1)D_{\alpha}(p \circ \cW\| p\otimes \sigma  ) \right)-s.
\end{align}
Using Young's inequality (Lemma \ref{lem-young}) in the form $\sup_{s\ge 0} s^{1-\alpha}c-s= \kappa_{\alpha  }c^{\frac{1}{\alpha}}$, where $\kappa_{\alpha} = {\alpha} (1-\alpha)^{\frac{1-\alpha}{\alpha}}\le 1$ we get
\begin{align}
\eps^{\rm{NS}}(M, \cW\otimes \cI_2) &\le   \inf_{p\in \cP(\cX)} \sup_{\sigma \in \cS(\cH),\; s\ge 0}  s^{1-\alpha} M^{1-\alpha} \exp\left((\alpha-1)D_{\alpha}(p \circ \cW\| p\otimes \sigma  ) \right)-s
\\&\le \inf_{p\in \cP(\cX)} \sup_{\sigma \in \cS(\cH)}   M^{\tfrac{1-\alpha}{\alpha}} \exp\left( \frac{\alpha-1}{\alpha} D_{\alpha}(p \circ \cW\| p\otimes \sigma  ) \right).
\end{align}
Taking $M = e^{nr}$, $n$ iid copies of the channel we deduce that
\begin{align}
\eps^{\rm{NS}}(e^{nr}, \cW^{\otimes n} \otimes \cI_2)
&\le \inf_{p \in \cP(\cX^n)} \sup_{\sigma^n \in \cS(\cH^{\otimes n})}  e^{\frac{nr(1-\alpha)}{\alpha}}\exp\left( \frac{\alpha-1}{\alpha} D_{\alpha}(p\circ \cW^{\otimes n}\| p\otimes \sigma^n  ) \right)
\\&\le  \inf_{p \in \cP(\cX)}   e^{\frac{nr(1-\alpha)}{\alpha}}\exp\left( \frac{\alpha-1}{\alpha} \inf_{\sigma^n \in \cS(\cH^{\otimes n})}D_{\alpha}(p^{\times n}\circ \cW^{\otimes n}\|p^{\otimes n}\otimes \sigma^n  ) \right)
\\&= \inf_{p\in \cP(\cX)}   e^{\frac{nr(1-\alpha)}{\alpha}}\exp\left( n\frac{\alpha-1}{\alpha} \inf_{\sigma \in \cS(\cH^{})}D_{\alpha}(p\circ \cW\|p\otimes \sigma  ) \right),
\end{align}
where we use  the additivity of the R\'enyi mutual information \cite[Lemma 7]{Hayashi2016Oct} (Lemma \ref{app-lem-additivity}) in the last equality. Taking the logarithm on both sides we obtain the desired achievability bound
\begin{align}
\frac{1}{n}\log\eps^{\rm{NS}}(e^{nr}, \cW^{\otimes n}\otimes \cI_2) &\le   - \sup_{p\in \cP(\cX)}\inf_{\sigma\in \cS(\cH)}    \frac{1-\alpha}{\alpha} (D_{\alpha}(p \circ \cW\| p\otimes  \sigma)-r). 
\end{align}
\end{proof}


\subsection{Converse}
\label{sec-NS-exp-cvs}

In this section we prove the converse result of Theorem \ref{prop-NS-exponent} which is known in the literature under the name of sphere-packing bound (SPB) \cite{Dalai2013Sep,Cheng2019Jan}. To this end, we prove a lower bound on the activated non-signalling error probability.

\begin{prop}\label{prop-exp-cvs}
Let $\cW$ be a classical-quantum channel with capacity $C(\cW)$ and zero R\'enyi capacity $C_0(\cW)$, and $ r\in (C_0(\cW), C(\cW))$. Then, there exists a constant $A>0$ and an integer $N_0$ depending only in $r$ such that for all $n\ge N_0$,
\begin{align}
\eps^{\rm{NS}}(e^{nr}, \cW^{\otimes n}\otimes \cI_2 ) \ge     \inf_{\alpha \in(0,1]}  \inf_{p\in \cP(\cX)} \sup_{\sigma \in \cS(\cH)} \exp(-A\sqrt{n})\exp\Big( -  n\frac{1-\alpha}{\alpha} [\exs{x\sim p}{D_{\alpha}( W_x\| \sigma )}-r]\Big).
\end{align}
This further implies that
\begin{align}
E^{\rm{ANS}}(r)\le  \sup_{\alpha \in(0,1]} \sup_{p\in \cP(\cX)} \inf_{\sigma \in \cS(\cH)} \frac{1-\alpha }{\alpha}\Big( \exs{x\sim p}{D_{\alpha}( W_x\| \sigma )}-r\Big). 
\end{align}
\end{prop}

In order to prove this result we first show a reduction to probability distributions uniformly supported on one type. This lemma is similar to the reduction  to constant composition codes for plain strategies.\footnote{For plain coding, a similar reduction is done in the literature (see e.g., \cite{Csiszar2014Jul,Dalai2013Sep,Cheng2019Jan}), so that one can focus on proving the SPB on constant composition codes\,--\,that consist of codewords with the same type.}

\begin{lem}\label{lem:constant-composition}
Let $\cX$ be a finite alphabet and  $\cW = \{W_x\}_{x\in \cX}$ be a classical-quantum channel. Let $n, M$ be integers and $ r\in (C_0(\cW), C(\cW))$. Let $\cT_n(\cX)$ be the set of types of alphabet $\cX$ and length $n$ and denote $\gamma_n = \tfrac{1}{2}(n+1)^{-|\cX|} $. The activated non-signaling error probability for the channel $\cW$ is lower bounded as
\begin{align}
\eps^{\rm{NS}}(M, \cW^{\otimes n} \otimes \cI_2)  &\ge   \inf_{T\in \cT_n(\cX)} \sup_{\sigma \in \cS(\cH) ,\, s\ge 0} \gamma_n\tr{W_{T}^{\otimes n}\wedge Ms\sigma^{\otimes n}}-s,
\end{align}
where $W_{T}^{\otimes n} = \bigotimes_{x\in \cX}W_x^{\otimes nT_x}$.  
\end{lem}

\begin{proof}
From Proposition \ref{prop-dual} we have that
\begin{align}
\eps^{\rm{NS}}(M, \cW^{\otimes n} \otimes \cI_2) &= \inf_{p\in \cP(\cX)} \sup_{B\mge 0} \sum_{x \in \cX} p(x) \tr{W_x\wedge \tfrac{M}{2}B}-\tr{B}
\\&=  \inf_{p\in \cP(\cX)} \sup_{\sigma \in \cS(\cH) , s\ge 0} \sum_{x \in \cX} p(x) \tr{W_x\wedge \tfrac{M}{2}s\sigma}-s,
\end{align}
where we write $B\mge 0$ as $B= s\sigma$ with $s=\tr B\ge 0$ and $\sigma \in \cS(\cH)$ in the last equality. Now, in the i.i.d.\ setting we consider $n$ uses of the channel $\cW$ and want to send $M=e^{nr}$ messages. Using the notation $x_1^n = x_1 x_2 \cdots x_n$ and $W_{x_1^n}^{\otimes n} =  W_{x_1} \otimes W_{x_2} \otimes  \cdots \otimes W_{x_n}$, we have that 
\begin{align}
\eps^{\rm{NS}}( M, \cW^{\otimes n}\otimes \cI_2)  &=\inf_{p\in \cP(\cX^n)} \sup_{\sigma^n \in \cS(\cH^{\otimes n}), \,  s\ge 0} \sum_{x_1^n\in \cX^n} p(x_1^n) \tr{W_{x_1^n}^{\otimes n}\wedge \tfrac{M}{2}s\sigma^n}-s
\\ &\ge \inf_{p\in \cP(\cX^n)} \sup_{\sigma \in \cS(\cH), \, s\ge 0} \sum_{x_1^n\in \cX^n} p(x_1^n) \frac{1}{2}\tr{W_{x_1^n}^{\otimes n}\wedge Ms\sigma^{\otimes n}}-s.
\end{align}
The objective function in $p$ of the last expression is invariant under the permutation of $X_1 X_2\cdots X_n$ and convex in $p$ (as a supremum of affine functions in $p$) so we may assume that the optimal  $p$ is permutation invariant (Lemma \ref{lem-perm-inv}).  Let $\cT_n(\cX)$ be the set of types of alphabet $\cX$ and length $n$ (we refer to Appendix \ref{sec-types} for a brief introduction to the method of types). A permutation invariant probability distribution $p$ can be written as $p=\sum_{T\in \cT_n(\cX)} \alpha_T \cU_T$, where $(\alpha_T)_{T\in \cT_n(\cX)}$ is a probability distribution on $\cT_n(\cX)$ and $\cU_T$ is the uniform probability distribution supported on $T$. Since $|\cT_n(\cX)|\le (n+1)^{|\cX|}$, there is a type $T(p)$ such that $\alpha_{T(p)} \ge (n+1)^{-|\cX|} $. Hence using the notation $W_{T}^{\otimes n}=W_{x_1^n}^{\otimes n} $ for an arbitrary $x_1^n$ of type $T$, we have that
\begin{align}
\eps^{\rm{NS}}( M, \cW^{\otimes n}\otimes \cI_2)   &\ge \inf_{p\in \cP(\cX^n)} \sup_{\sigma \in \cS(\cH),\, s\ge 0} \sum_{T\in \cT_n(\cX)} \alpha_T \frac{1}{2}\tr{W_{T}^{\otimes n}\wedge Ms\sigma^{\otimes n}}-s
\\&\ge   \inf_{p\in \cP(\cX^n)} \sup_{\sigma \in \cS(\cH), \, s\ge 0}  (n+1)^{-|\cX|}  \frac{1}{2}  \tr{W_{T(p)}^{\otimes n}\wedge Ms\sigma^{\otimes n}}-s
\\&\ge  \inf_{p\in \cP(\cX^n)}  \inf_{T\in \cT_n(\cX)} \sup_{\sigma \in \cS(\cH),\, s\ge 0} \gamma_n \tr{W_{T}^{\otimes n}\wedge Ms\sigma^{\otimes n}}-s
\\&=   \inf_{T\in \cT_n(\cX)} \sup_{\sigma \in \cS(\cH) ,\, s\ge 0} \gamma_n\tr{W_{T}^{\otimes n}\wedge Ms\sigma^{\otimes n}}-s \label{eq-s},
\end{align}
where $\gamma_n = \tfrac{1}{2} (n+1)^{-|\cX|}$. 
\end{proof}

Note that once we reduce the problem to minimizing over probability distributions uniformly supported on one type (Lemma \ref{lem:constant-composition}), the proofs of \cite{Dalai2013Sep,Cheng2019Jan} apply even though they were claimed for only plain strategies. This is because these proofs are based on bounding the worst case error probability $\eps_{\max}$ by a hypothesis testing divergence (see e.g., \cite[Proposition 10]{Cheng2019Jan}), which can be related to the meta-converse for any coding scheme $\cC=(\cE, \cD)$ as\footnote{Translating a converse bound for the worst case error probability to a converse bound on the average case error probability is standard in coding theory.}
\begin{align}
\eps_{\max}(\cC, M, \cW) =\sup_{x\in \cE} \eps(x)  \ge \sup_{\sigma\in \cS(\cH)}\inf_{x \in \cE} \inf_{0\mle O \mle \dI} \Big\{\tr{W_x(\dI-O)} \;\Big|\; \tr{\sigma O}\le \tfrac{1}{M}\Big\}.
\end{align}

In Appendix \ref{app-spb}, we nevertheless provide a self-contained proof of the SPB (Proposition \ref{prop-exp-cvs}), using the dual perspective \eqref{ANS-dual}, based on a Blahut method \cite{Blahut1974Jul} and results from \cite{Audenaert2008Apr,Cheng2019Jan}. This proof has the advantage of not relying on large deviation techniques and could be of independent interest.


\section{On activation and the need for a critical rate}
\label{sec:improved rounding}

\subsection{Approach}

In this section we prove an achievability bound on the non-signaling error probability without activation. In contrast to the setting where activation is allowed (Proposition \ref{prop-exp-ach}), we are only able to get the error exponent termed with the Petz R\'enyi divergence under some assumptions: $(a)$ coding over channels with positive non-signaling assisted zero-error capacity or $(b)$ coding over general channels above a certain critical rate. The first observation is that if we can use a sub-linear number of copies of the channel to send noiselessly one bit then activation is not needed for error exponents as it can be simulated without incurring a significant loss in the rate transmission. We discuss this observation in more details in Section \ref{sec:zero-error}. The second attempt to go beyond the activation assumption uses a rounding result on the non-activated non-signalling error probability, as well as an inequality between Petz and sandwiched R\'enyi divergences (Section \ref{sec:critical-rate}).  


\subsection{Coding over channels with positive zero-error capacity}\label{sec:zero-error}

The non-signaling capability $M^{\rm{NS}}(\cW)$ of a classical-quantum channel $\cW$ is the largest number of messages that can be transmitted through a channel with zero-error and when non-signaling correlations are available \cite{Duan2015Dec}. The non-signaling zero-error classical capacity is then defined as the optimal rate of the transmitted information over iid channels with non-signaling strategies achieving zero-error
\begin{equation}
C^{\rm{NS}}_0(\cW) = \lim_{n\rightarrow \infty} \frac{1}{n}\log M^{\rm{NS}}(\cW^{\otimes n}).
\end{equation}
The zero-error non-signalling assisted capacity is pleasantly quantified by $C^{\rm{NS}}_0(\cW)=C_0(\cW)$ \cite{Duan2015Dec}. In this section, we consider a channel $\cW$ with positive zero-error classical capacity and show that, in this case, activation is not needed when considering non-signaling error exponents.

\begin{prop}
Let $\cW$ be a classical-quantum channel with $C_0(\cW)>0$, and $r\ge 0$. Then, we have that
\begin{align}
E^{\rm{NS}}(r)= E^{\rm{ANS}}(r). 
\end{align}
\end{prop}

\begin{proof}
If $\cW$ is a channel with positive zero-error capacity then there is $n_0\in \mathbb{N}$ such that $M^{\rm{NS}}(\cW^{\otimes n_0})\ge 2$ which means $n_0$ copies of the channel $\cW$, assisted by non-signaling correlations, can be used to transmit one bit of information with zero error. Now, for an integer $n\ge n_0$, given a non-signaling super channel $\Pi^n$ that achieves the activated non-signaling error probability $\eps^{\rm{NS}}(e^{nr}, \cW^{\otimes n-n_0}\otimes \cI_2)$,  (see \eqref{eq:eps-NS-def}), i.e.
\begin{align}
1-\eps^{\rm{NS}}(e^{nr}, \cW^{\otimes n-n_0}\otimes \cI_2) &=   \frac{1}{e^{nr}} \sum_{m=1}^{e^{nr}} \bra{m}\Pi^n (\cW^{\otimes n-n_0}\otimes \cI_2) (\proj{m}) \ket{m}.
\end{align}
We want to construct a super channel $\widetilde{\Pi}^n$ that takes into input only iid copies of $\cW$, i.e., without the perfect one bit channel $\cI_2$. To this end, we use $n_0$ copies of $\cW$ to simulate $\cI_2$, we call this non-signaling super channel $\Pi^{n_0}_{\rm{sim}} : \cW^{\otimes n_0}\mapsto \cI_2$. We then have 
\begin{align}
1-\eps^{\rm{NS}}(e^{nr}, \cW^{\otimes n}) &\ge    \frac{1}{e^{nr}} \sum_{m=1}^{e^{nr}} \bra{m}\Pi^n\circ  (\id^{n-n_0}\otimes \Pi^{n_0}_{\rm{sim}}) (\cW^{\otimes n}) (\proj{m}) \ket{m}
\\&=  \frac{1}{e^{nr}} \sum_{m=1}^{e^{nr}} \bra{m}\Pi^n (\cW^{\otimes n-n_0}\otimes \cI_2) (\proj{m}) \ket{m}
\\&= 1-\eps^{\rm{NS}}(e^{nr}, \cW^{\otimes n-n_0}\otimes \cI_2).
\end{align}
Hence, we get
\begin{align}
\eps^{\rm{NS}}(e^{nr}, \cW^{\otimes n}\otimes \cI_2)\le \eps^{\rm{NS}}(e^{nr}, \cW^{\otimes n})\le \eps^{\rm{NS}}(e^{nr}, \cW^{\otimes n-n_0}\otimes \cI_2).
\end{align}
From Proposition \ref{prop-exp-ach},  we have that 
\begin{align}
\eps^{\rm{NS}}(e^{nr}, \cW^{\otimes n}) &\le\eps^{\rm{NS}}(e^{nr}, \cW^{\otimes n-n_0}\otimes \cI_2)
\\&\le \inf_{\alpha \in(0,1]}  \inf_{p\in \cP(\cX)}\sup_{\sigma\in \cS(\cH)}   \exp\Big( nr\frac{1-\alpha}{\alpha}-(n-n_0)\frac{1-\alpha}{\alpha} \exs{x\sim p}{D_{\alpha}( W_x\| \sigma )}\Big). \label{eq:ach-zero-error-cap}
\end{align}
Similarly, from Proposition  \ref{prop-exp-cvs} we have that for a constant $A>0$ that
\begin{align}
\eps^{\rm{NS}}(e^{nr}, \cW^{\otimes n})&\ge  \eps^{\rm{NS}}(e^{nr}, \cW^{\otimes n}\otimes \cI_2 ) \\&\ge  \inf_{\alpha \in(0,1]}  \inf_{p\in \cP(\cX)} \sup_{\sigma \in \cS(\cH)} \exp(-
A\sqrt{n})\exp\Big( -  n\frac{1-\alpha}{\alpha} [\exs{x\sim p}{D_{\alpha}( W_x\| \sigma )}-r]\Big). \label{eq:cvs-zero-error-cap}
\end{align}
Finally, from \eqref{eq:ach-zero-error-cap} and \eqref{eq:cvs-zero-error-cap} we conclude that for a channel $\cW$ with positive zero-error capacity
\begin{align}
\lim_{n\rightarrow \infty } - \frac{1}{n}\log \eps^{\rm{NS}}(e^{nr}, \cW^{\otimes n})  = \sup_{\alpha \in(0,1]} \sup_{p\in \cP(\cX)} \inf_{ \sigma\in \cS(\cH)} \frac{1-\alpha}{\alpha}\left( \exs{x\sim p}{D_{\alpha}( W_x\| \sigma )}-r\right). 
\end{align}
\end{proof}


\subsection{Coding above a critical rate}
\label{sec:critical-rate}

In this section, we show that activation is not necessary for optimal coding with non-signaling correlations above a certain critical rate.

\begin{prop}\label{prop-exp-ach-w/-activ}
Let $\cW$ be a classical-quantum channel and $M\in\mathbb{N}$. Then, we have that
\begin{align}
\eps^{\rm{NS}}(M,\cW)&\leq\inf_{\alpha\in (0,1)} \inf_{p\in \cP(\cX)} \sup_{ \sigma\in \cS(\cH)}\Bigg\{(2M)^{\frac{1-\alpha}{\alpha}} \exp\left( -\frac{1-\alpha}{\alpha}D_{\alpha}(p\circ \cW\| p\otimes \sigma)\right)\nonumber
\\&\qquad\qquad\qquad\qquad\qquad\quad + \frac{2}{M}\exp\left( -(1-\alpha)D_{\alpha}(p\circ \cW\| p\otimes \sigma)\right)\Bigg\}.\label{eq:second-part}
\end{align}
Moreover, this implies for $r > r'_c:= \sup_{\alpha\in (0,1)}  (1-\alpha)^2 C_\alpha(\cW)$ that
\begin{align}\label{eq:NS=ANS-above-rate}
E^{\rm{NS}}(r)=E^{\rm{ANS}}(r).
\end{align}
\end{prop}

The recently found plain coding result \cite{Dalai2013Sep,Renes2024Jul,Li2024Jul} as stated in \eqref{EE-plain} provides the same expansion above a possibly different critical rate. In general, it is not clear whether the new critical rate improves on the old one \eqref{eq:critical-old} of \cite{Dalai2013Sep,Renes2024Jul,Li2024Jul}. In the particular case of classical-quantum channels outputting pure states we show in Corollary \ref{cor:pure-ee} that the new critical rate is actually below the capacity of order $0$ and thus there is no critical rate with non-signaling assistance (as opposed to the plain setting). Moreover, classically, there is no critical rate for coding with non-signaling correlations \cite{Polyanskiy2012Dec}.

For the proof of Proposition \ref{prop-exp-ach-w/-activ} we need a modified analysis compared to the case of activated non-signaling codes as treated in Proposition \ref{prop-exp-ach}.

\begin{lem}\label{lem:rounding-MC-one-shot}
Let $\cW$ be a classical-quantum channel and $M\in\mathbb{N}$. Then, we have that
\begin{align}
\eps^{\rm{NS}}(M,\cW)&\le  \inf_{p\in \cP(\cX)} \sup_{B=B^{\dagger}}\Bigg\{\left(\sum_{x\in \cX} p(x) \tr{W_x\wedge MB^+}-\frac{1}{2}\tr{B^+}\right)\nonumber\\
&\qquad\qquad\qquad\quad+\left(\frac{1}{M}\sum_{x\in \cX} p(x)  \tr{ (W_x+MB^-) \wedge MB^+} -\frac{1}{2}\tr{B^+}\right)\Bigg\}.\label{eq:extra-term-cq}
\end{align}
\end{lem}

\begin{proof}
From Proposition \ref{prop-dual}, we have that 
\begin{align}
\eps_{}^{\rm{NS}}(M, \cW)  &=  \inf_{p\in \cP(\cX)}  \sup_{B=B^\dagger} \sum_{x\in \cX} p(x) \tr{ W_x \wedge  MB }-\tr{B}.
\end{align}
The difficulty is now that the supremum is over $B=B^\dagger$ and not $B\mge0$ as in the case of activated non-signaling codes (Proposition \ref{prop-dual}). We estimate
\begin{align}
&M\eps_{}^{\rm{NS}}(M,\cW) 
\\&=  \inf_{p\in \cP(\cX)} \sup_{B=B^\dagger} \sum_{x\in \cX} p(x) M\tr{W_x\wedge MB}-M\tr{B}
\\&= \inf_{p\in \cP(\cX)} \sup_{B=B^\dagger} \sum_{x\in \cX} p(x) (M-1)\tr{W_x\wedge MB}-M\tr{B^+}+M\tr{B^-}+ \tr{W_x\wedge MB}
\\&\overset{(a)}{\le} \inf_{p\in \cP(\cX)} \sup_{B=B^\dagger} \sum_{x\in \cX} p(x) (M-1)\tr{W_x\wedge MB}-M\tr{B^+}+\tr{(W_x+MB^-)\wedge (MB+MB^-)}
\\&\overset{(b)}{\le}  \inf_{p\in \cP(\cX)} \sup_{B=B^\dagger} \sum_{x\in \cX} p(x) M\tr{W_x\wedge MB^+}-M\tr{B^+}+\tr{(W_x+MB^-)\wedge MB^+},
\end{align}
where in $(a)$ we apply the joint concavity for the non-commutative minimum \cite{Cheng2023Nov} (Lemma \ref{lem:concavity of min}); in $(b)$ we used the fact that $B\mle B^+$ and $B+B^-=B^+$.
\end{proof}

Lemma \ref{lem:rounding-MC-one-shot} again immediately recovers the known result $\eps^{\rm{NS}}(M,\cW)= \eps^{\rm{MC}}(M,\cW)$ for classical channels $\cW$, since we have $\tr{(W_x+MB^-)\wedge MB^+}= \tr{W_x\wedge MB^+}$. However, for the classical-quantum case more work is needed to properly control the second term in \eqref{eq:extra-term-cq}\,---\,and this is ultimately where we again pick up the critical rate constraint.\\

The following Lemma \ref{lem:sand-vs-petz} gives a reversed Petz-sandwiched R\'enyi divergences inequality that will be useful to bound the second term in \eqref{eq:extra-term-cq} in terms of Petz R\'enyi divergences. It was first proven for $\alpha>1$ in \cite[Corollary 3.6]{Jencova2018Aug} and then extended to $\alpha\geq1/2$ in \cite[Proposition 11]{Wilde2018Aug}. For completeness, we give a self-contained proof based on H\"older's inequality in Appendix \ref{app:reverse-petz-sandwiched}, where we also make the connection to the pretty good bounds from \cite{Iten2016Dec}.

\begin{lem}\label{lem:sand-vs-petz}
Let $\rho \in \cS(\cH)$, $\sigma \mge 0$, and $\alpha\in [\frac{1}{2}, \infty)$. Then, we have that
\begin{align}\label{eq:petz-sand}
D_{2-1/\alpha}( \rho\| \sigma)\leq\widetilde{D}_{\alpha}( \rho\| \sigma).
\end{align}
\end{lem}

Employing Lemma \ref{lem:rounding-MC-one-shot} together with Lemma \ref{lem:sand-vs-petz}, we are finally able to prove Proposition \ref{prop-exp-ach-w/-activ} on the achievability of the non-signaling assisted error exponent.

\begin{proof}[Proof of Proposition \ref{prop-exp-ach-w/-activ}]
Recall the notation $p \circ \cW = \sum_{x \in \cX} p(x)\proj{x}\otimes W_x$ and $p = \sum_{x \in \cX} p(x)\proj{x}$. Using Lemma~\ref{lem:rounding-MC-one-shot} and writing $B^+\mge 0$ as $B^+=s\sigma$, where $s=\tr{B^+}\ge 0$ and $\sigma\in \cS(\cH)$ and $\tau=B^-\mge0$ with $\tau\sigma=\sigma\tau=0$, we have for $\alpha, \beta  \in (0,1)$ that
\begin{align}
\eps^{\rm{NS}}(M,\cW)
&\le  
\inf_{p\in \cP(\cX)} \sup_{\substack{\sigma\in \cS(\cH), \, s\ge 0, \, \tau\mge 0\\\tau\sigma=\sigma\tau=0}  } \bigg\{\left(\textstyle \sum_{x\in \cX} p(x) \tr{W_x\wedge sM\sigma} - \frac{s}{2}\right)\nonumber\\
&\qquad\qquad\qquad\qquad\qquad\qquad\qquad+\left(\tfrac{1}{M}\textstyle \sum_{x\in \cX} p(x)  \tr{ (W_x +M\tau) \wedge sM \sigma } - \frac{s}{2}\right)\bigg\}
\\&\overset{(a)}{=}\inf_{p\in \cP(\cX)}\sup_{\substack{\sigma\in \cS(\cH), \, s\ge 0, \, \tau\mge 0\\\tau\sigma=\sigma\tau=0}}\bigg\{ \left(\tr{\textstyle \sum_{x\in \cX} p(x)\proj{x}\otimes W_x\wedge \textstyle \sum_{x\in \cX} p(x)\proj{x}\otimes 2sM\sigma} - s\right)\nonumber\\
&\qquad\qquad\qquad\qquad+\left(\tfrac{1}{M}\tr{ \textstyle \sum_{x\in \cX} p(x)\proj{x}\otimes (W_x+M\tau) \wedge \textstyle \sum_{x\in \cX} p(x)\proj{x}\otimes 2sM \sigma } - s\right)\bigg\}
\\&\overset{(b)}{\le} \inf_{p\in \cP(\cX)}\sup_{\substack{\sigma\in \cS(\cH),  \, \tau\mge 0\\\tau\sigma=\sigma\tau=0}}\bigg\{\sup_{s\ge 0}\Big\{\left(s^{1-\alpha}(2M)^{1-\alpha}\tr{(p\circ \cW)^{\alpha} (p\otimes \sigma)^{1-\alpha}} - s\right)\Big\}\nonumber\\
&\qquad \qquad\qquad\qquad  +\sup_{s\ge 0}\left\{\left(\tfrac{1}{M}s^{1-\beta}(2M)^{1-\beta}  \tr{ \left(p\circ \cW+Mp\otimes \tau\right)^{\beta}  (p\otimes \sigma)^{1-\beta} } - s\right)\right\}\bigg\}
\\&\overset{(c)}{\le} \inf_{p\in \cP(\cX)} \sup_{\substack{\sigma\in \cS(\cH),  \, \tau\mge 0\\\tau\sigma=\sigma\tau=0}}\bigg\{(2M)^{\frac{1-\alpha}{\alpha}}\tr{(p\circ \cW)^{\alpha} (p\otimes \sigma)^{1-\alpha}}^{\frac{1}{\alpha}}\nonumber\\
&\qquad\qquad\qquad\qquad\qquad\qquad\qquad+\frac{2^{\frac{1-\beta}{\beta}}}{M}\tr{ \left(p\circ \cW+Mp\otimes \tau\right)^{\beta}  (p\otimes \sigma)^{1-\beta} }^{\frac{1}{\beta}}\bigg\}, \label{eq:beta}
\end{align}
where in $(a)$ we use the block diagonal decomposition of the non-commutative minimum; in $(b)$ we use Audenaert's inequality as stated in \eqref{Chernoff}; and in $(c)$ we use Young's inequality in the form (Lemma \ref{lem-young})
\begin{align}
\text{$\sup_{s\ge 0} s^{1-\alpha}c-s= \kappa_{\alpha  }c^{\frac{1}{\alpha}}$ with $\kappa_{\alpha} = {\alpha} (1-\alpha)^{\frac{1-\alpha}{\alpha}}\le 1$.}
\end{align}
Moreover, we have $\tr{(p\circ \cW)^{\alpha} (p\otimes \sigma)^{1-\alpha}} = \exp\left(  -(1-\alpha){D}_{\alpha}(p\circ \cW\| p\otimes \sigma)\right)$. For $\beta \ge \frac{1}{2}$, the $\beta$-dependent term can be bounded  as follows
\begin{align}
&\tr{ \left(p\circ \cW+Mp\otimes \tau)\right)^{\beta}  (p\otimes \sigma)^{1-\beta} }^{\frac{1}{\beta}}
\\&\overset{(a)}{\le} \tr{  \left((p\otimes \sigma)^{(1-\beta)/(2\beta)}\left(p\circ \cW+Mp\otimes \tau\right)  (p\otimes \sigma)^{(1-\beta)/(2\beta)}\right)^{\beta} }^{\frac{1}{\beta}}
\\&\overset{(b)}{=} \tr{  \left((p\otimes \sigma)^{(1-\beta)/(2\beta)}(p\circ \cW)  (p\otimes \sigma)^{(1-\beta)/(2\beta)}\right)^{\beta} }^{\frac{1}{\beta}}
\\&= \exp\left( -\frac{1-\beta}{\beta}\widetilde{D}_{\beta}\left(p\circ \cW\middle\|p\otimes \sigma\right) \right) \label{eq:before-sand-Petz}
\\&\overset{(c)}{\le} \exp\left( -\frac{1-\beta}{\beta}D_{2-1/\beta}\left(p\circ \cW\middle\|p\otimes \sigma\right) \right),
\end{align}
where in $(a)$ we used the Araki-Lieb-Thirring inequality \cite{Araki1990Feb,Lieb}; in $(b)$ we used $\sigma\tau= \tau \sigma =0$; in $(c)$ we employ the reversed Petz-sandwiched R\'enyi inequality from Lemma \ref{lem:sand-vs-petz}.
 Choosing $\beta=\frac{1}{2-\alpha}\ge \frac{1}{2}$ then gives the claimed overall upper bound
\begin{align}
\eps^{\rm{NS}}(M,\cW)&\leq\inf_{\alpha\in (0,1)} \inf_{p\in \cP(\cX)} \sup_{ \sigma\in \cS(\cH)}\Bigg\{(2M)^{\frac{1-\alpha}{\alpha}} \exp\left( -\frac{1-\alpha}{\alpha}D_{\alpha}(p\circ \cW\| p\otimes \sigma)\right)\nonumber
\\&\qquad\qquad\qquad\qquad\qquad\quad + \frac{2}{M}\exp\left( -(1-\alpha)D_{\alpha}(p\circ \cW\| p\otimes \sigma)\right)\Bigg\}.\label{eq:one-shot-bound-critical-rate}
\end{align}
We note that we can use the reverse H\"older type  inequality $D_\alpha(\rho\| \sigma) \le \frac{1}{\alpha}\widetilde{D}_\alpha(\rho\| \sigma)$ of \cite{Iten2016Dec} to bound \eqref{eq:before-sand-Petz}. Define the critical rate
\begin{align}\label{eq:critical-rate}
r'_c:= \sup_{\alpha\in (0,1)}  (1-\alpha)^2 \sup_{p\in \cP(\cX)}\inf_{ \sigma\in \cS(\cH)} D_{\alpha}(p\circ \cW\| p\otimes \sigma),
\end{align}
and we then  have that for all $r > r'_c$ that
\begin{align}
\eps^{\rm{NS}}(e^{nr},\cW^{\otimes n})&\overset{(a)}{\leq} \inf_{\alpha\in (0,1)} \inf_{p\in \cP(\cX)}\Bigg\{\sup_{ \sigma^n\in \cS(\cH^{\otimes n})}2^{\frac{1-\alpha}{\alpha}} \exp\left(nr\frac{1-\alpha}{\alpha} -\frac{1-\alpha}{\alpha}D_{\alpha}(p^{\otimes n}\circ \cW^{\otimes n}\| p^{\otimes n}\otimes \sigma^n)\right) \nonumber
\\&\qquad\qquad\qquad\qquad +\sup_{ \sigma^n\in \cS(\cH^{\otimes n})} 2\exp\left(-nr -(1-\alpha)D_{\alpha}(p^{\otimes n}\circ \cW^{\otimes n}\| p^{\otimes n}\otimes \sigma^n)\right)\Bigg\}
\\&\overset{(b)}{\leq}\inf_{\alpha\in (0,1)} \inf_{p\in \cP(\cX)}\Bigg\{ \sup_{ \sigma\in \cS(\cH)}2^{\frac{1-\alpha}{\alpha}} \exp\left(nr\frac{1-\alpha}{\alpha} -n\frac{1-\alpha}{\alpha}D_{\alpha}(p\circ \cW\| p\otimes \sigma)\right) \nonumber
\\&\qquad\qquad\qquad\qquad +\sup_{ \sigma\in \cS(\cH)} 2\exp\left(-nr -n(1-\alpha)D_{\alpha}(p\circ \cW\| p\otimes \sigma)\right)\Bigg\}
\\&\overset{(c)}{\leq} \inf_{\alpha\in (0,1)} \inf_{p\in \cP(\cX)}\Bigg\{ \sup_{ \sigma\in \cS(\cH)}(2^{\frac{1-\alpha}{\alpha}}+2) \exp\left(nr\frac{1-\alpha}{\alpha} -n\frac{1-\alpha}{\alpha}D_{\alpha}(p\circ \cW\| p\otimes \sigma)\right)\Bigg\},\label{eq:alpha-opt}
\end{align}
where in $(a)$ we set $\cW\leftarrow \cW^{\otimes n}$, $M=e^{nr}$ and choose $p^n=p^{\otimes n}$ in \eqref{eq:one-shot-bound-critical-rate}; in $(b)$ we  the additivity of the R\'enyi mutual information \cite[Lemma 7]{Hayashi2016Oct} (Lemma \ref{app-lem-additivity}); in $(c)$ we use the definition of the critical rate in \eqref{eq:critical-rate}. Finally, as $r> C_0(\cW)$, optimal $\alpha$ in \eqref{eq:alpha-opt} is bounded below by a positive constant independent of $n$ thus the prefactor $(2^{\frac{1-\alpha}{\alpha}}+2)$ is finite and we obtain the asymptotic lower bound on the  error exponent:
\begin{align}
\lim_{n\rightarrow \infty } - \frac{1}{n}\log\eps^{\rm{NS}}(e^{nr},\cW^{\otimes n})&\ge  \sup_{\alpha\in (0,1)} \sup_{p\in \cP(\cX)} \inf_{ \sigma\in \cS(\cH)} \frac{1-\alpha}{\alpha}(D_{\alpha}(p\circ \cW\| p\otimes \sigma) -r)
\\&= \sup_{\alpha \in(0,1]}  \frac{1-\alpha}{\alpha}\left( C_{\alpha}(\cW)-r\right). 
\end{align}
\sloppy The other direction needed to prove \eqref{eq:NS=ANS-above-rate} follows from Proposition \ref{prop-exp-cvs} as $\eps^{\rm{NS}}(e^{nr},\cW^{\otimes n})\ge \eps^{\rm{NS}}(e^{nr},\cW^{\otimes n}\otimes \cI_2)$. 
\end{proof}

Proposition \ref{prop-exp-ach-w/-activ} is not tight in general as we know that there is no critical rate in the classical setting \cite{Polyanskiy2012Dec}.\footnote{We lose optimality when we apply the reversed Petz-sandwiched R\'enyi inequality from Lemma \ref{lem:sand-vs-petz} to bound \eqref{eq:before-sand-Petz}. For the binary symmetric channel, the new critical rate $r'_c$ from \eqref{eq:critical-rate} is numerically strictly below the old critical rate $r_c$ from \eqref{eq:critical-old}. However, $r'_c$ is strictly larger than the zero R\'enyi capacity for this example.} As an application of Proposition \ref{prop-exp-ach-w/-activ}, we establish the non-signaling error exponent of channels that solely output pure states. In contrast, without non-signaling assistance, the error exponents of such channels are only known above the critical rate \eqref{eq:critical-old} \cite{Hayashi2007Dec}\&\cite[Appendix A]{Beigi2023Oct}.

\begin{cor}\label{cor:pure-ee}
Let $\cW = \{W_x\}_{x\in \cX}$ be a classical-quantum channel such that the output states $\{W_x\}_{x\in \cX}$ are pure, and $r\geq0$. Then, we have that
\begin{align}\label{eq:ee-pure}
E^{\rm{NS}}(r)=E^{\rm{ANS}}(r).
\end{align}
\end{cor}

\begin{proof}
By Proposition \ref{prop-exp-ach-w/-activ} the claimed \eqref{eq:ee-pure} is satisfied for all $r> \sup_{\alpha\in (0,1)} (1-\alpha)^2 C_\alpha(\cW)$. The corollary follows by showing that $\sup_{\alpha\in (0,1)}(1-\alpha)^2 C_\alpha(\cW)\le C_0(\cW)$. To this end, we define for $p\in \cP(\cX)$:
\begin{align}
C_\alpha(p, \cW) = \inf_{ \sigma\in \cS(\cH)} D_{\alpha}(p\circ \cW\| p\otimes \sigma).
\end{align}
This function is related to the auxiliary function and  has a closed form expression by the quantum Sibson's identity \cite{Sharma2013Feb}
\begin{align}
C_\alpha(p, \cW) = - \frac{\alpha}{1-\alpha}\log\tr{\left(\textstyle\sum_{x\in \cX} p_x W_x^{\alpha}\right)^{1/\alpha}}.
\end{align}
When the output states $\{W_x\}_{x\in \cX}$ are pure we have that 
\begin{align}
(1-\alpha)C_\alpha(p, \cW) = - \alpha \log\tr{\left(\textstyle\sum_{x\in \cX} p_x W_x\right)^{1/\alpha}} = -\log \left\| \sum_{x\in \cX} p_x W_x \right\|_{\frac{1}{\alpha}} , \label{eq:nondec}
\end{align}
where $\| \cdot \|_p$ denotes the Schatten $p$-norm. The function in~\eqref{eq:nondec} is thus non-increasing due to the anti-monotonicity of the Schatten-norms in their parameter. 
Therefore,
\begin{align}
\sup_{\alpha\in (0,1)}  (1-\alpha)^2 C_\alpha(\cW)\le \sup_{\alpha\in (0,1)} \sup_{p\in \cP(\cX)}  (1-\alpha) C_\alpha(p, \cW) = \sup_{p\in \cP(\cX)}   C_0(p, \cW) = C_0(\cW). 
\end{align}
\end{proof}
Note that in above proof, we crucially used the fact that the output states are pure to show that $\alpha\mapsto (1-\alpha)C_\alpha(p, \cW)$ is non-increasing on $(0, \infty)$. In general, however, this function is not monotone.


\section{Non-signaling assisted quantum error exponents}
\label{sec:NS-coding-QQ}

In Section \ref{sec:NS-coding-CQ} we focused on non-signaling coding error exponents in the classical-quantum setting. We observe that the achievability proof can be extended to the fully quantum setting with  similar proofs. To this end, we first state a generalization of Proposition \ref{prop-dual} that gives  the dual formulation of the meta-converse and activated non-signaling error probabilities in the quantum setting. 

\begin{prop}\label{prop-dual-quantum}
Let $\cW_{A\rightarrow B}$ be a quantum channel with Choi matrix
\begin{align}
(J_{\cW})_{RB}=\sum_{i,j=1}^{|A|}\ket{i}\bra{j}_R\otimes \cW_{A\rightarrow B}(\ket{i}\bra{j}_A).
\end{align}
The one-shot meta-converse and activated non-signaling error probabilities satisfy
\begin{align}\label{NS-dual-quantum}
\eps_{}^{\rm{NS}}(2M,\cW\otimes \cI_2) = \eps_{}^{\rm{MC}}(M,\cW) &= \inf_{\rho_R \in \cS(R)} \sup_{Z_B\mge 0}\Big\{\tr{\sqrt{\rho_R} (J_{\cW})_{RB} \sqrt{\rho_R}\wedge M\rho_R\otimes Z_B}-\tr{Z_B}\Big\}. 
\end{align}
\end{prop}

Using this formulation, we obtain the achievability bound of the activated non-signalling error probability in the quantum setting.

\begin{prop}\label{prop-exp-ach-quantum}
Let $\cW_{A\rightarrow B}$ be a quantum channel with Choi matrix $(J_{\cW})_{RB}$, $r\ge 0$, and $\alpha\in(0,1]$. Then, we have that
\begin{align}
E^{\rm{ANS}}(r)\ge  \sup_{\rho_R \in \cS(R)}\inf_{\sigma_B\in \cS(B)}    \frac{1-\alpha}{\alpha} \Big(D_{\alpha}(\sqrt{\rho_R} (J_{\cW})_{RB} \sqrt{\rho_R}\| \rho_R\otimes  \sigma_B)-r\Big). 
\end{align}
\end{prop}

The proof is similar to the proofs of Propositions \ref{prop-dual} \& \ref{prop-exp-ach}, and is included in Appendix \ref{app:fully-quantum} for completeness. Similar to Section \ref{sec:improved rounding}, one could achieve the same performance as in Proposition \ref{prop-exp-ach-quantum} without activation for channels with positive zero-error non-signaling classical capacity or above a critical rate. On the other hand, the converse result (sphere-packing bound) is hard to generalize to the fully quantum setting even for plain coding. The main obstacle for this generalization is the minimization over arbitrary quantum input that can be non product. This is similar to the difficulty faced in the problem of quantum channel discrimination with adaptive strategies (see, e.g., \cite{Cooney2016Jun}).


\section{Conclusion}
\label{sec:conclusion}

We demonstrated that the (activated) non-signaling error exponents for coding over classical-quantum channels are determined by the sphere packing bound with Petz-R\'enyi divergence for all rates, analogous to the classical case. Moreover we showed that activation is not needed for channels with  positive zero-error non-signaling capacity, channels outputting pure states, and for all channels above a critical rate.  Future research directions include $(a)$ investigating the necessity of activation for optimal non-signaling error exponents, 
$(b)$ generalizing the sphere packing bound to the fully quantum setting, and $(c)$ designing rounding protocols similar to the ones in \cite{barman2017algorithmic,Fawzi19,Oufkir24} on channel coding, or \cite{Berta24,Cao24,Yao24} on channel simulation, or \cite{Fawzi24,Ferme24} on network coding problems, with the goal of relating the sphere packing bound with plain, shared randomness, or shared entanglement coding strategies.


\section*{Acknowledgments}
We would like to thank the Erd\H{o}s Center for organizing The Focused Workshop on Quantum R\'enyi Divergences, where part of this work was carried out. We also thank Mark M. Wilde for pointing us to the references \cite{Jencova2018Aug} and \cite{Wilde2018Aug}. 
\\MB and AO acknowledge funding by the European Research Council (ERC Grant Agreement No. 948139), MB acknowledges support from the Excellence Cluster - Matter and Light for Quantum Computing (ML4Q). MT is supported by the National Research Foundation, Singapore and A*STAR under its CQT Bridging Grant.

\printbibliography


\appendix 

\section{Method of types}
\label{sec-types}

Let $n$ be an integer and  $\cX$ be a finite alphabet of size $|\cX|$. Let $x_1^n = x_1\cdots x_n$ be an element of $\cX^n$. For $x\in\cX$ we define $n(x|x_1^n)$ to be the number of occurrences of $x$ in the sequence $x_1\cdots x_n$
\begin{equation}
n(x|x_1^n) = \sum_{t=1}^n \bid\{x_t=x\}. 
\end{equation}
A type $T$ is a probability distribution on $\cX$ of the form 
\begin{equation}
T=\Big\{\frac{n_x}{n}\Big\}_{x\in \cX} \text{ where } n_x\in \mathbb{N} \text{ and } \sum_{x\in \cX} n_x=n. \label{eq:type}
\end{equation}
The set of types of alphabet $\cX$ of length $n$ is denoted $\cT_n(\cX)$. It is a finite set of size satisfying 
\begin{equation}
|\cT_n(\cX)|\le (n+1)^{|\cX|}. 
\end{equation}
This simple bound can be proved using the simple  observation that each $n_x$ in \eqref{eq:type} satisfies $n_x \in \{0,1, \dots, n\}$ and thus it has at most $n+1$ possibilities. We say that $x_1^n = x_1\cdots x_n$ has type $T$ and write $x_1^n \sim T$ if for all $x\in \cX$ we have $ {n(x| x_1^n)} = {n}T_x$. A probability distribution $p$ on $\cX^n$ is permutation invariant if for all permutation $\sigma \in \fS_n$, for all $ x_1\cdots x_n \in \cX^n$ we have 
\begin{equation}
p(x_1\cdots x_n) = p(x_{\sigma_1}\cdots x_{\sigma_n}).
\end{equation}
Let  $T\in \cT_n(\cX)$ be a type and  $p$ be a permutation invariant  probability distribution on $\cX^n$.
Clearly if two sequences $x_{1}^n$ and $y_{1}^n$ have the same type $T$ then $y_{1}^n$ can be obtained by permuting the elements of $x_{1}^n$ so
\begin{equation}
p(x_1^n)= p(y_1^n) , \quad \forall x_1^n\sim T, \; \forall y_1^n\sim T.
\end{equation}
We denote this value by $\frac{\alpha_T}{|T|}$ where $|T|$ is the number of elements in $\cX^n$ of type $T$. So we can write 
\begin{equation}
p = \sum_{T\in \cT_n(\cX)} \alpha_T \cU_T,
\end{equation}
where $\cU_T$ is the uniform probability distribution supported on $T$:
\begin{equation}
\cU_T(x_1^n) = \frac{1}{|T|}\bid\{x_1^n\sim T\}.
\end{equation}
Since $p$ and $\{\cU_T\}_{T\in \cT_n(\cX)}$ are all probability distributions on $\cX^n$, $(\alpha_T)_{T\in \cT_n(\cX)}$ is a probability distribution on $\cT_n(\cX)$. In particular, since we have $|\cT_n(\cX)|\le (n+1)^{|\cX|}$ there is a type $T^\star$ such that 
\begin{equation}
\alpha_{T^\star} \ge \frac{1}{ (n+1)^{|\cX|}}
\end{equation}
because otherwise $\sum_{T\in \cT_n(\cX)}\alpha_T< \sum_{T\in \cT_n(\cX)} \frac{1}{ (n+1)^{|\cX|}} = |\cT_n(\cX)|\cdot\frac{1}{ (n+1)^{|\cX|}} \le 1$ which contradicts the fact that $(\alpha_T)_{T\in \cT_n(\cX)}$ is a probability distribution on $\cT_n(\cX)$.


\section{Alternative proof of the sphere packing bound}
\label{app-spb}

Here restate and give a self-contained proof of Proposition \ref{prop-exp-cvs}.

\begin{prop}[Restatement of Proposition \ref{prop-exp-cvs}]\label{prop-exp-cvs-app}
Let $\cW$ be a classical-quantum channel with capacity $C(\cW)$ and zero R\'enyi capacity $C_0(\cW)$, and $ r\in (C_0(\cW), C(\cW))$. Then, there exists a constant $A>0$ and an integer $N_0$ depending only in $r$ such that for all $n\ge N_0$,
\begin{align}
\eps^{\rm{NS}}(e^{nr}, \cW^{\otimes n}\otimes \cI_2 ) \ge     \inf_{\alpha \in(0,1]}  \inf_{p\in \cP(\cX)} \sup_{\sigma \in \cS(\cH)} \exp(-A\sqrt{n})\exp\Big( -  n\frac{1-\alpha}{\alpha} [\exs{x\sim p}{D_{\alpha}( W_x\| \sigma )}-r]\Big).
\end{align}
This further implies that
\begin{align}
E^{\rm{ANS}}(r)\le  \sup_{\alpha \in(0,1]} \sup_{p\in \cP(\cX)} \inf_{\sigma \in \cS(\cH)} \frac{1-\alpha }{\alpha}\left( \exs{x\sim p}{D_{\alpha}( W_x\| \sigma )}-r\right). 
\end{align}
\end{prop}

\begin{proof}
Let  $\sigma \in \cS(\cH)$ and consider the eigenvalue decomposition of $W_x= \sum_{i=1}^{d} \lambda^x_i \proj{\phi_i^x}$ and $\sigma= \sum_{j=1}^{d} \mu_j \proj{\psi_j}$. Let $(p^x,q^x)\in  \cP([d]\times [d])^2$ be the Nussbaum-Szko\l a probability distributions for the pair of states $(W_x, \sigma)$ defined as:
\begin{align}\label{eq:NS}
p^{x}_{i,j} = \lambda^x_i  |\spr{\phi^x_i}{\psi_j}|^2 \; \text{ and }\; q^{x}_{i,j}  = \mu_j |\spr{\phi^x_i}{\psi_j}|^2, \quad \forall (i,j)\in [d]\times [d].
\end{align}
Let $T \in  \cT_n(\cX)$ and for $x\in \cX$ denote by $n(x|T) = nT_x$ the number of $x$ in the type $T$. We have $W_T^{\otimes n} = \bigotimes_{x\in \cX} W_x^{\otimes n(x|T)}$ so if 
$(p^T,q^T)$ are the Nussbaum-Szko\l a probability distributions for the pair of states $(W_T^{\otimes n}, \sigma^{\otimes n})$ we have that 
\begin{align}
p^T= \bigotimes_{x\in \cX} (p^x)^{\otimes n(x|T)} \;\text{ and }\; q^T = \bigotimes_{x\in \cX} (q^x)^{\otimes n(x|T)}.
\end{align}
For a set of probability distributions $v=\{v^{x}\}_{x\in \cX}\in \cP([d]\times [d])^{\cX}$ (defined later in the proof), we denote 
\begin{align}
v^T= \bigotimes_{x\in \cX} (v^x)^{\otimes n(x|T)}.  
\end{align}
Setting $\pi_0=\pi_1=1$ in \cite[Proposition $2$]{Audenaert2008Apr} and using the identity $\tr{A\wedge B} = \inf_{0\mle O\mle \dI} \tr{AO} +\tr{B(\dI-O)}$ we get for any integer $M$ and any real number $s\ge 0$
\begin{align}
\tr{W_T^{\otimes n}\wedge Ms \sigma^{\otimes n} }\ge \frac{1}{2} \sum_{i,j}  \min\left( p^{T}_{i,j} , Ms q^{T}_{i,j} \right).
\end{align}
Hence, using Lemma \ref{lem:constant-composition}, we obtain the lower bound\footnote{This step is similar to Nagaoka's method \cite{Nagaoka2006Nov}: One can obtain Nagaoka's result by choosing $s$ carefully in \eqref{ng}.}
\begin{align}
\eps^{\rm{NS}}(M, \cW^{\otimes n}\otimes \cI_2) 
&\ge \inf_{T\in \cT_n(\cX)}\sup_{\sigma \in \cS(\cH),\, s\ge 0}  \gamma_n\tr{W_T^{\otimes n}\wedge Ms\sigma^{\otimes n}}-s
\\&\ge  \inf_{T\in \cT_n(\cX)} \sup_{\sigma \in \cS(\cH) ,\, s\ge 0} \frac{1}{2}  \sum_{i,j}  \gamma_n\min\left( p^{T}_{i,j} , Ms q^{T}_{i,j} \right)-s \label{ng}
\\&=  \inf_{T\in \cT_n(\cX)} \sup_{\sigma \in \cS(\cH),\,  s\ge 0} \frac{\gamma_n}{2} \sup_{v\in \cP([d]\times [d])^{\cX}}\sum_{i,j} v^T_{i,j} \min\bigg( \frac{p^{T}_{i,j} }{v^T_{i,j} }, Ms\cdot  \frac{q^{T}_{i,j} }{v^T_{i,j} }\bigg)-s.
\end{align}
For a set of probability distributions $v=\{v^{x}\}_{x\in \cX}$, define the good event 
\begin{align}
\cG^{T}_{v,p}=\left\{(i,j) : \log\left(\tfrac{p^{T}_{i,j} }{v^{T}_{i,j} } \right)\ge -n\exs{x\sim T}{D(v^{x}\|p^{x})}-\sqrt{4n \exs{x\sim T}{ \var(v^{x}\|p^x)} } \right\}.
\end{align}
By Chebyshev's inequality we have 
\begin{align}\label{Chebyshev 1}
\sum_{ i,j } v^{T}_{i,j} \bid\{ (i,j)\in \cG^{T}_{v,p}\}  \ge \frac{3}{4}.
\end{align}
Similarly, we define the good event related to $q$:
\begin{align}
\cG^{T}_{v,q}=\left\{(i,j) : \log\left(\tfrac{q^{T}_{i,j} }{v^{T}_{i,j} } \right)\ge - n\exs{x\sim T}{D(v^{x}\|q^{x})}-\sqrt{4n\exs{x\sim T}{ \var(v^{x}\|q^x)}} \right\}
\end{align}
Again, by Chebyshev's inequality we have 
\begin{align}\label{Chebyshev 2}
\sum_{ i,j } v^{T}_{i,j} \bid\{ (i,j)\in \cG^{T}_{v,q}\}\ge \frac{3}{4}.
\end{align}
We refer to Appendix \ref{app-Chebyshev} for detailed proofs of these inequalities.\footnote{We note that this idea was used by Blahut \cite{Blahut1974Jul} to prove the classical SPB.} By the union bound, we obtain from \eqref{Chebyshev 1} and \eqref{Chebyshev 2}
\begin{align}
\sum_{ i,j } v^{T}_{i,j}  \bid\{(i,j)\in \cG^{T}_{v,p}\cap\cG^{T}_{v,q} \} \ge \frac{1}{2}.
\end{align}
Therefore the activated non-signalling error probability can be lower bounded as follows:
\begin{align}
&\eps^{\rm{NS}}(M, \cW^{\otimes n}\otimes \cI_2) 
\\&\ge  \inf_{T\in \cT_n(\cX)} \sup_{\sigma \in \cS(\cH), \, s\ge 0} \frac{\gamma_n}{2} \sup_{v\in \cP([d]\times [d])^{\cX}} \sum_{i,j} v^T_{i,j} \min\bigg( \frac{p^{T}_{i,j} }{v^T_{i,j} }, Ms\cdot  \frac{q^{T}_{i,j} }{v^T_{i,j} }\bigg)-s
\\&\ge  \inf_{T\in \cT_n(\cX)} \sup_{\sigma \in \cS(\cH), \, s\ge 0} \frac{\gamma_n}{2} \sup_{v\in \cP([d]\times [d])^{\cX}}\sum_{(i,j)\in \cG^{T}_{v,p}\cap\cG^{T}_{v,q}} v^T_{i,j} \exp\min\bigg( \log\frac{p^{T}_{i,j} }{v^T_{i,j} }, \log(Ms)+\log  \frac{q^{T}_{i,j} }{v^T_{i,j} }\bigg)-s
\\&
\begin{aligned}
\ge  \inf_{T\in \cT_n(\cX)} \sup_{\sigma \in \cS(\cH) ,\, s\ge 0} \frac{\gamma_n}{4} \sup_{v\in \cP([d]\times [d])^{\cX}} \exp\min\Big(  - n\exs{x\sim T}{D(v^{x}\|p^{x})}-\sqrt{4n \exs{x\sim T}{ \var(v^{x}\|p^x)}},
\\ \log(Ms)  - n\exs{x\sim T}{D(v^{x}\|q^{x})}-\sqrt{4n \exs{x\sim T}{ \var(v^{x}\|q^x)}}\Big)-s. 
\end{aligned}\label{eq-min}
\end{align}
In the following we use the techniques of \cite{Cheng2019Jan}. \cite[Proposition 3]{Cheng2019Jan} proved the following saddle-point result
\begin{align}\label{saddle}
\sup_{\sigma\in \cS(\cH)}\inf_{0<\alpha \le 1} \frac{1-\alpha}{\alpha} (r-   \exs{x\sim T}{D_{\alpha}(W_x\|\sigma)})=   \inf_{0<\alpha \le 1}\sup_{\sigma\in \cS(\cH)} \frac{1-\alpha}{\alpha} (r-   \exs{x\sim T}{D_{\alpha}(W_x\|\sigma)}).
\end{align}
We denote by $(\alpha_T^\eta, \sigma_T^\eta)$ an optimizer of \eqref{saddle} for a rate $r'=r+\eta$ slightly larger than $r$, $\eta$ will be chosen later so that $\eta \rightarrow 0$ when $n\rightarrow \infty$.

For this specific state $\sigma = \sigma_T^\eta$, we consider $(p^x,q^x)$ the Nussbaum-Szko\l a probability distributions for the pair of states $(W_x, \sigma_T^\eta)$ defined in \eqref{eq:NS}. 
Then, we choose $v = v^{\alpha_T^\eta}$ to be the tilted probability distribution of parameter $\alpha_T^{\eta}$
\begin{align}\label{eq:tilted}
v^x_{i,j} = v^{x, \alpha_T^\eta}_{i,j} = \frac{(p^x_{i,j})^{\alpha_T^\eta}(q^x_{i,j})^{1-\alpha_T^\eta}}{\sum_{i',j'} (p^x_{i',j'})^{\alpha_T^\eta}(q^x_{i',j'})^{1-\alpha_T^\eta}}, \quad \forall x\in \cX.
\end{align}
Moreover since $(\alpha_T^\eta, \sigma_T^\eta)$ is an optimizer of \eqref{saddle} for the rate $r'=r+\eta$ we have that \cite[Lemma 16]{Cheng2019Jan} 
\begin{align}
\exs{x\sim T}{D(v^{x, \alpha_T^\eta}\|p^{x})}&= \frac{1-\alpha_T^\eta}{\alpha_T^\eta} [\exs{x\sim T}{D_{\alpha}(W_x\|\sigma_T^\eta)}-r'], \\
\exs{x\sim T}{D(v^{x, \alpha_T^\eta}\|q^{x})}&=r'.
\end{align}
Besides, for the choice of $\sigma = \sigma_T^\eta$ and $\alpha = \alpha_T^\eta$ the variance terms in \eqref{eq-min} are uniformly bounded \cite[Appendix C]{Cheng2019Jan}. Therefore by setting $M=e^{nr}$ and  choosing $v=v^{\alpha_T}$ from \eqref{eq:tilted} and a constant $A>0$ depending only on the channel $\cW$ and on the rate $r$ \cite[Eq. $(44)$]{Cheng2019Jan}, we obtain
\begin{align}
&\eps^{\rm{NS}}(e^{nr}, \cW^{\otimes n}\otimes \cI_2) 
\\& \begin{aligned}
\ge \inf_{T\in \cT_n(\cX)} \sup_{s\ge 0} \frac{\gamma_n}{4} \sup_{v\in \cP([d]\times [d])^{\cX}} \exp\min\Big(  - n\exs{x\sim T}{D(v^{x}\|p^{x})}-\sqrt{4n \exs{x\sim T}{ \var(v^{x}\|p^x)}},
\\\log(Ms)  - n\exs{x\sim T}{D(v^{x}\|q^{x})}-\sqrt{4n \exs{x\sim T}{ \var(v^{x}\|q^x)}}\Big)-s
\end{aligned}
\\&\ge   \inf_{T\in \cT_n(\cX)} \sup_{ s\ge 0} \exp(-A\sqrt{n})  \exp\min\Big(  - n\tfrac{1-\alpha_T^\eta}{\alpha_T^\eta} \big[\mathbb{E}_{x\sim T} \big[D_{\alpha_T^\eta}(W_x\|\sigma_T^\eta)\big]-r-\eta\big],\log(s) -n\eta \Big)-s. 
\end{align}
We choose $s$ such that the arguments of the last minimum function are equal, i.e.,
\begin{align}
\log(s) -n\eta =  - n\tfrac{1-\alpha_T^\eta}{\alpha_T^\eta} \big[\mathbb{E}_{x\sim T} \big[D_{\alpha_T^\eta}(W_x\|\sigma_T^\eta)\big]-r-\eta\big]
\end{align}
and obtain
\begin{align}
&\eps^{\rm{NS}}(e^{nr}, \cW^{\otimes n}\otimes \cI_2)  
\\&\ge   \inf_{T\in \cT_n(\cX)} \exp(-A\sqrt{n})  \exp\min\Big(   - n\tfrac{1-\alpha_T^\eta}{\alpha_T^\eta} \big[\mathbb{E}_{x\sim T} \big[D_{\alpha_T^\eta}(W_x\|\sigma_T^\eta)\big]-r-\eta\big],\log(s) -n\eta \Big)-s
\\&  \begin{aligned}
= \inf_{T\in \cT_n(\cX)}  \exp(-A\sqrt{n})  \exp\Big(   - n\tfrac{1-\alpha_T^\eta}{\alpha_T^\eta} \big[\mathbb{E}_{x\sim T} \big[D_{\alpha_T^\eta}(W_x\|\sigma_T^\eta)\big]-r-\eta\big] \Big)
\\ -\exp\Big(n\eta   - n\tfrac{1-\alpha_T^\eta}{\alpha_T^\eta} \big[\mathbb{E}_{x\sim T} \big[D_{\alpha_T^\eta}(W_x\|\sigma_T^\eta)\big]-r-\eta\big]\Big)
\end{aligned}
\\&\overset{(a)}{\ge}    \inf_{T\in \cT_n(\cX)}  \frac{1}{2}\exp(-A\sqrt{n})  \exp\Big(  -  n\tfrac{1-\alpha_T^\eta}{\alpha_T^\eta} \big[\mathbb{E}_{x\sim T} \big[D_{\alpha_T^\eta}(W_x\|\sigma_T^\eta)\big]-r-\eta\big] \Big)
\\&\overset{(b)}{=}      \inf_{T\in \cT_n(\cX)} \inf_{0<\alpha\le 1} \sup_{\sigma\in \cS(\cH)} \frac{1}{2}\exp(-A\sqrt{n})  \exp\Big(  - n\frac{1-\alpha}{\alpha} [\exs{x\sim T}{D_{\alpha}(W_x\|\sigma)}-r] -2\sqrt{n}A\frac{1-\alpha}{\alpha} \Big), \label{eq-alpha}
\end{align}
where in $(a)$ we choose $\eta = - \frac{2A}{\sqrt{n}}$ and in $(b)$ we use \eqref{saddle}. The additional term $2A\frac{1-\alpha}{\alpha} \sqrt{n}$ can make the lower bound useless if optimal $\alpha$ in \eqref{eq-alpha} is close to $0$. However, in the following we argue that for rates $r> C_0(\cW)$ the optimal $\alpha$ in \eqref{eq-alpha} cannot be arbitrary small. \footnote{A similar step in  \cite[Eq $(47)$]{Cheng2019Jan} uses the monotonicity and convexity of the function $r\mapsto\sup_{0< \alpha \le 1} \inf_{\sigma\in \cS(\cH)}  \frac{1-\alpha}{\alpha}[\exs{x\sim T}{D_{\alpha}(W_x\|\sigma)}-r].$}

For $ C(\cW)>r> C_0(\cW)= \sup_{p\in \cP(\cX)}\inf_{\sigma\in \cS(\cH)}\exs{x\sim p}{D_{0}(W_x\|\sigma)}$, by the monotonicity of $\alpha\mapsto C_\alpha(\cW)$, we can choose $\alpha'\in (0,1)$ such that  $r> C_{\alpha}(\cW)$ if and only if $\alpha < \alpha'$  and a sufficiently large $n\ge 4A^2 (r-C_{\alpha'}(\cW))^{-2} +  4A^2 (1-\alpha')^{-2}$ such that 
for all $\alpha < \alpha'$, for all $T\in \cT_n(\cX)$ we have that:
\begin{align}
&  \frac{1}{2}\exp(-A\sqrt{n})  \exp\Big(  - n\frac{1-\alpha}{\alpha} \big[\inf_{\sigma\in \cS(\cH)}\exs{x\sim T}{D_{\alpha}(W_x\|\sigma)}-r\big] -2\sqrt{n}A\frac{1-\alpha}{\alpha} \Big)
\\&\overset{(a)}{\ge}  \frac{1}{2}\exp(-A\sqrt{n})  \exp\Big(  - n\frac{1-\alpha}{\alpha} [C_{\alpha}(\cW)-r] -2\sqrt{n}A\frac{1-\alpha}{\alpha} \Big)
\\&\overset{(b)}{\ge}   \frac{1}{2}\exp(-A\sqrt{n})  \exp\Big( \frac{1-\alpha'}{\alpha'}  \big( n [r-C_{\alpha'}(\cW)] -2\sqrt{n}A \big)\Big)
\\&\overset{(c)}{\ge}   \frac{1}{2}\exp(-A\sqrt{n})  \exp\Big( n\frac{1-\alpha'}{\alpha'}  \Big)\overset{(d)}{\ge} 1,
\end{align}
where we used in $(a)$   $\inf_{\sigma\in \cS(\cH)}\exs{x\sim T}{D_{\alpha}(W_x\|\sigma)}\le C_{\alpha}(\cW)$; in $(b)$ the fact that $\alpha \mapsto C_{\alpha}(\cW)$ (resp. $\alpha \mapsto \frac{1-\alpha}{\alpha}$) is non-decreasing \cite[Proposition $2$ $(h)$]{Cheng2019Jan} (resp.~decreasing); in $(c)$   $n\ge 4A^2 (r-C_{\alpha'}(\cW))^{-2}$  and in $(d)$  $n\ge 4A^2 (1-\alpha')^{-2}$.  In this case, the optimal $\alpha$ in \eqref{eq-alpha} should satisfy $\alpha \ge  \alpha'$ since otherwise the right hand side of \eqref{eq-alpha} will be at least $1$ which contradicts $ \eps^{\rm{NS}}(M, \cW^{\otimes n}) \le 1$. Therefore by restricting $\alpha \ge \alpha'$ and relaxing the minimization over $p$ to be on the set of probability distributions we obtain from \eqref{eq-alpha}
\begin{align}
& \eps^{\rm{NS}}(e^{nr}, \cW^{\otimes n} \otimes \cI_2) 
\\&\ge   \inf_{p\in \cP(\cX)} \inf_{\alpha'
\le \alpha\le 1} \sup_{\sigma\in \cS(\cH)} \frac{1}{2}\exp(-A\sqrt{n})  \exp\Big(  - n\frac{1-\alpha}{\alpha} [\exs{x\sim T}{D_{\alpha}(W_x\|\sigma)}-r] -2\sqrt{n}A\frac{1-\alpha}{\alpha} \Big) 
\\&\ge  \inf_{p\in \cP(\cX)} \inf_{\alpha'
\le \alpha\le 1} \sup_{\sigma\in \cS(\cH)} \frac{1}{2}\exp(-A\sqrt{n})  \exp\Big(  - n\frac{1-\alpha}{\alpha} [\exs{x\sim T}{D_{\alpha}(W_x\|\sigma)}-r] -2\sqrt{n}A\frac{1-\alpha'}{\alpha'} \Big) 
\\ &=    \inf_{p\in \cP(\cX)} \inf_{\alpha'
\le\alpha\le 1} \sup_{\sigma\in \cS(\cH)} \exp(-A'\sqrt{n})  \exp\Big(  - n\frac{1-\alpha}{\alpha} [\exs{x\sim p}{D_{\alpha}(W_x\|\sigma)}-r]  \Big) 
\end{align}
for a constant $A' = A+ 2A\frac{1-\alpha'}{\alpha'}$ depending only on the channel $\cW$ and on the rate $r$. Finally, we deduce the desired asymptotic converse bound
\begin{align}
\lim_{n\rightarrow\infty} -\frac{1}{n}\log\eps^{\rm{NS}}(e^{nr}, \cW^{\otimes n}\otimes \cI_2) 
&\le   \sup_{p\in \cP(\cX)} \sup_{\alpha' \le \alpha \le 1} \inf_{\sigma\in \cS(\cH)} \frac{1-\alpha}{\alpha} \left(\exs{x\sim P}{D_{\alpha}( W_x\|  \sigma)}-r\right)
\\&\le   \sup_{\alpha\in (0, 1]} \sup_{p\in \cP(\cX)}\inf_{\sigma\in \cS(\cH)}   \frac{1-\alpha}{\alpha} \left(\exs{x\sim P}{D_{\alpha}( W_x\|  \sigma)}-r\right).
\end{align}
\end{proof}


\section{Chebyshev approximations}
\label{app-Chebyshev}

We only prove \eqref{Chebyshev 1}, with \eqref{Chebyshev 2} being similar. Let $r=\{v^{x}\}_x$ be a set of probability distributions,  recall the definition of the good set 
\begin{align}
\cG^{T}_{v,p}=\left\{(i,j) : \frac{1}{n}\log\left(\frac{p^{T}_{i,j} }{v^{T}_{i,j} } \right)\ge - \exs{x\sim T}{D(v^{x}\|p^{x})}-\sqrt{\frac{1}{4n} \exs{x\sim T}{ \var(v^{x}\|p^x)} } \right\}.
\end{align}
Since we have independence:
\begin{align}
p^T= \bigotimes_{x\in \cX} (p^x)^{\otimes n(x|T)} \text{ and } r^T = \bigotimes_{x\in \cX} (v^{x})^{\otimes n(x|T)},
\end{align}
the first order can be computed exactly 
\begin{align}
\sum_{ i,j } v^{T}_{i,j} \frac{1}{n}\log\left(\frac{p^{T}_{i,j} }{v^{T}_{i,j} } \right) &= \sum_{ i_1^n, j_1^n } \prod_{s=1}^n {r^{x_s}_{i_s,j_s} } \frac{1}{n}\log\left(\prod_{t=1}^n \frac{p^{x_t}_{i_t,j_t} }{r^{x_t}_{i_t,j_t} } \right)
\\&= \sum_{t=1}^n\sum_{ i_t, j_t } r^{x_t}_{i_t, j_t} \frac{1}{n}\log\left( \frac{p^{x_t}_{i_t,j_t} }{r^{x_t}_{i_t,j_t} } \right)
\\&= \sum_{x\in \cX} \frac{n(x|T)}{n} \sum_{ i, j } v^{x}_{i, j} \log\left( \frac{p^{x}_{i,j} }{v^{x}_{i,j} } \right)
\\&= - \exs{x\sim T}{D(v^{x}\|p^{x})}. 
\end{align}
Similarly, using independence, the second order can be written as:
\begin{align}
&\sum_{ i,j } v^{T}_{i,j} \left(\frac{1}{n}\log\left(\frac{p^{T}_{i,j} }{v^{T}_{i,j} } \right) +\exs{x\sim T}{D(v^{x}\|p^{x})}\right)^2 
\\&= \sum_{ i,j } v^{T}_{i,j} \left(\frac{1}{n} \sum_{t=1}^n \log\left(\frac{p^{x_t}_{i_t,j_t} }{r^{x_t}_{i_t,j_t} } \right) +D(r^{x_t}\|p^{x_t})\right)^2 
\\&=  \sum_{ i,j } v^{T}_{i,j} \frac{1}{n^2} \sum_{t=1}^n \left(\log\left(\frac{p^{x_t}_{i_t,j_t} }{r^{x_t}_{i_t,j_t} } \right) +D(r^{x_t}\|p^{x_t})\right)^2 
\\&= \frac{1}{n^2} \sum_{t=1}^n  \sum_{ i_t,j_t } r^{x_t}_{i_t,j_t}  \left(\log\left(\frac{p^{x_t}_{i_t,j_t} }{r^{x_t}_{i_t,j_t} } \right) +D(r^{x_t}\|p^{x_t})\right)^2 
\\&= \frac{1}{n} \sum_{x\in \cX} \frac{n(x|T)}{n}  \sum_{ i,j } v^{x}_{i,j}  \left(\log\left(\frac{p^{x}_{i,j} }{v^{x}_{i,j} } \right) +D(v^{x}\|p^{x})\right)^2 
\\&= \frac{1}{n} \exs{x\sim T}{\var(v^{x}\|p^{x})}. 
\end{align}
By Chebyshev's inequality we have 
\begin{align}
&\sum_{ i,j } v^{T}_{i,j} \bid\{ (i,j)\in \overline{\cG^{T}_{v,p}}\}
\\&=\sum_{ i,j } r^{T}_{i,j}\left\{ \frac{1}{n}\log\left(\tfrac{p^{T}_{i,j} }{v^{T}_{i,j} } \right)< - \exs{x\sim T}{D(v^{x}\|p^{x})}-\sqrt{\tfrac{4}{n} \exs{x\sim T}{ \var(v^{x}\|p^x)} } \right\}
\\&\le \frac{\tfrac{1}{n} \exs{x\sim T}{\var(v^{x}\|p^{x})}}{\tfrac{4}{n} \exs{x\sim T}{ \var(v^{x}\|p^x)}} = \frac{1}{4}.
\end{align}


\section{Reverse Petz-sandwiched R\'enyi inequality}
\label{app:reverse-petz-sandwiched}

Here, we prove the full version of Lemma \ref{lem:sand-vs-petz} along with new inequalities.

\begin{lem}[Reformulated Lemma \ref{lem:sand-vs-petz}]\label{app:sand-vs-petz}
 We have that 
\begin{align}
\frac{1}{\alpha}\widetilde{D}_{\alpha}( \rho\| \sigma)  &\ge  \widetilde{D}_{\frac{1}{2-\alpha}}( \rho\| \sigma), && \forall \rho\mge 0,\, \forall  \sigma \in \cS(\cH),\, \forall \alpha\in (0,1], 
\\\widetilde{D}_{\frac{1}{2-\alpha}}( \rho\| \sigma) &\ge D_{\alpha}( \rho\| \sigma), && \forall \rho\in \cS(\cH),\, \forall \sigma \mge 0,\,  \forall \alpha\in [0,2].  \label{eq:rev-petz-sand}
\end{align}
\end{lem}
In particular, for $\alpha\in (0,1)$ and $\rho, \sigma \in \cS(\cH)$ we have that 
\begin{align}
    \frac{1}{\alpha}\widetilde{D}_{\alpha}( \rho\| \sigma)  &\ge  \widetilde{D}_{\frac{1}{2-\alpha}}( \rho\| \sigma)\ge D_{\alpha}( \rho\| \sigma)
\end{align}
recovering the reverse inequality of \cite{Iten2016Dec}. 

The inequality \eqref{eq:rev-petz-sand} was proven for $\alpha \in (1,2)$ in \cite[Corollary 3.6]{Jencova2018Aug} and then extended to $\alpha \in [0,2)$ in \cite[Proposition 11]{Wilde2018Aug}. In the following we provide a self contained proof. 
\begin{proof}
To prove $\frac{1}{\alpha}\widetilde{D}_{\alpha}( \rho\| \sigma)\ge \widetilde{D}_{\frac{1}{2-\alpha}}( \rho\| \sigma)$ we can use \cite[Proposition 4.12]{Tomamichel2016}
\begin{align}
\widetilde{D}_{\alpha}( \rho\| \sigma) = \lim_{n\rightarrow \infty} \frac{1}{n} {D}_{\alpha}( \mathcal{P}_{\sigma^{\otimes n}}(\rho^{\otimes n}) \| \sigma^{\otimes n})
\end{align}
and prove that the inequality  $\frac{1}{\alpha}\widetilde{D}_{\alpha}( \rho\| \sigma)\ge \widetilde{D}_{\frac{1}{2-\alpha}}( \rho\| \sigma)$ holds for commuting $[\rho, \sigma]=0$.  In the case where $[\rho, \sigma]=0$, we have 
\begin{align}
\widetilde{D}_{\frac{1}{2-\alpha}}( \rho\| \sigma) &= -\frac{2-\alpha}{1-\alpha}\log \tr{\rho^{\frac{1}{2-\alpha}} \sigma^{\frac{1-\alpha}{2-\alpha}}}
\\ \frac{1}{\alpha}\widetilde{D}_{\alpha}( \rho\| \sigma)&= \frac{1}{\alpha(\alpha-1)}\log \tr{\rho^{\alpha}\sigma^{1-\alpha}}
\end{align}
so in case of $\alpha\in (0,1)$,  $\frac{1}{\alpha}\widetilde{D}_{\alpha}( \rho\| \sigma)\ge \widetilde{D}_{\frac{1}{2-\alpha}}( \rho\| \sigma)$ is equivalent to 
\begin{align}
\tr{\rho^{\alpha}\sigma^{1-\alpha}}=\tr{\left( \rho^{\frac{1}{2-\alpha}} \sigma^{\frac{1-\alpha}{2-\alpha}}\right)^{\alpha(2-\alpha)} \sigma^{(1-\alpha)^2}} \le \left(\mathrm{Tr} \rho^{\frac{1}{2-\alpha}} \sigma^{\frac{1-\alpha}{2-\alpha}}\right)^{\alpha(2-\alpha)} \tr{\sigma}^{(1-\alpha)^2}
\end{align}
which is true by H\"older's inequality. 

We move to prove $\widetilde{D}_{\frac{1}{2-\alpha}}( \rho\| \sigma) \ge D_{\alpha}( \rho\| \sigma)$. We make cases depending on whether $\alpha\in(0,1)$ or $\alpha \in(1,2)$. Note that for $\alpha=1$ we have $ \widetilde{D}_{1}( \rho\| \sigma) =D( \rho\| \sigma)= D_{1}( \rho\| \sigma)$.
\paragraph{Case $\alpha\in(0,1)$.}
From the definitions of the sandwiched and Petz R\'enyi divergences, we have that 
\begin{align}
D_{\alpha}( \rho\| \sigma) &= -\frac{1}{1-\alpha}\log \tr{\rho^{\alpha}\sigma^{1-\alpha}},
\\ \widetilde{D}_{\frac{1}{2-\alpha}}( \rho\| \sigma) &= -\frac{2-\alpha}{1-\alpha}\log \mathrm{Tr} \left(\sigma^{\frac{1-\alpha}{2}}\rho \sigma^{\frac{1-\alpha}{2}}\right)^{\frac{1}{2-\alpha}}. 
\end{align}
Let $A= \sigma^{\frac{1-\alpha}{2}} \rho^{\frac{\alpha}{2}}$ and $B= \rho^{\frac{1-\alpha}{2}}$. Let $p=2$, $q=\frac{2}{1-\alpha}$ and $r= \frac{2}{2-\alpha}$. We have $\frac{1}{p}+\frac{1}{q}=\frac{1}{r}$ so by the generalized H\"older's inequality \cite{bhatia2013matrix} we have 
\begin{align}
\|AB\|_{r}\le \|A\|_p \|B\|_q. 
\end{align}
Observe that 
\begin{align}
\|AB\|_{r} &= \left(\tr{(AB^2A^\dagger)^{\frac{r}{2}}}\right)^{\frac{1}{r}} 
\\&= \left(\tr{\big( \sigma^{\frac{1-\alpha}{2}} \rho^{\frac{\alpha}{2}}(\rho^{\frac{1-\alpha}{2}})^2  \rho^{\frac{\alpha}{2}}\sigma^{\frac{1-\alpha}{2}}\big)^{\frac{1}{2-\alpha}}}\right)^{\frac{2-\alpha}{2}}
=\left(\tr{( \sigma^{\frac{1-\alpha}{2}} \rho\sigma^{\frac{1-\alpha}{2}})^{\frac{1}{2-\alpha}}}\right)^{\frac{2-\alpha}{2}}.
\end{align}
Moreover, we have 
\begin{align}
\|A\|_p &= \left(\tr{(AA^{\dagger})^{\frac{p}{2}}}\right)^{\frac{1}{p}} = \left(\tr{AA^{\dagger}}\right)^{\frac{1}{2}} =\left( \tr{\sigma^{\frac{1-\alpha}{2}} \rho^{\frac{\alpha}{2}} \rho^{\frac{\alpha}{2}}\sigma^{\frac{1-\alpha}{2}}}\right)^{\frac{1}{2}} = \left(\tr{\rho^{\alpha}\sigma^{1-\alpha}}\right)^{\frac{1}{2}},
\\ \|B\|_q &= \left(\tr{(BB^{\dagger})^{\frac{q}{2}}}\right)^{\frac{1}{q}} =\left(\tr{(\rho^{\frac{1-\alpha}{2}}\rho^{\frac{1-\alpha}{2}})^{\frac{1}{1-\alpha}}}\right)^{\frac{1-\alpha}{2}}= \left(\tr{\rho}\right)^{\frac{1-\alpha}{2}}=1.
\end{align}
Hence, we have
\begin{align}
\left(\tr{\rho^{\alpha}\sigma^{1-\alpha}}\right)^{\frac{1}{2}}&= \|A\|_p \|B\|_q \ge   \|AB\|_{r} = \left(\tr{( \sigma^{\frac{1-\alpha}{2}} \rho\sigma^{\frac{1-\alpha}{2}})^{\frac{1}{2-\alpha}}}\right)^{\frac{2-\alpha}{2}}.
\end{align}
Therefore by taking the logarithm in both sides then multiplying by $-\frac{2}{1-\alpha}$ we obtain the desired inequality
\begin{align}
{D}_{\alpha}( \rho\| \sigma) = -\frac{1}{1-\alpha}\log\tr{\rho^{\alpha}\sigma^{1-\alpha}}\le -\frac{2-\alpha}{1-\alpha}\log \tr{( \sigma^{\frac{1-\alpha}{2}} \rho\sigma^{\frac{1-\alpha}{2}})^{\frac{1}{2-\alpha}}}= \widetilde{D}_{\frac{1}{2-\alpha}}( \rho\| \sigma). 
\end{align}
\paragraph{Case $\alpha\in(1,2)$.}
From the definitions of the sandwiched and Petz R\'enyi divergences, we have that 
\begin{align}
D_{\alpha}( \rho\| \sigma) &= \frac{1}{\alpha-1}\log \tr{\rho^{\alpha}\sigma^{1-\alpha}},
\\ \widetilde{D}_{\frac{1}{2-\alpha}}( \rho\| \sigma) &= \frac{2-\alpha}{\alpha-1}\log \mathrm{Tr} \left(\sigma^{\frac{1-\alpha}{2}}\rho \sigma^{\frac{1-\alpha}{2}}\right)^{\frac{1}{2-\alpha}}. 
\end{align}
Let $A= \sigma^{\frac{1-\alpha}{2}}\rho^{\frac{1}{2}}$ and $B = \rho^{\frac{\alpha-1}{2}}$. Let $p = \frac{2}{2-\alpha}$, $q=\frac{2}{\alpha-1}$ and $r=2$ we have that $\frac{1}{p}+ \frac{1}{q}=\frac{1}{r}$ so 
by the generalized H\"older's inequality \cite{bhatia2013matrix} we have 
\begin{align}
\|AB\|_{r}\le \|A\|_p \|B\|_q. 
\end{align}
Observe that 
\begin{align}
\|AB\|_{r} &= \left(\tr{(ABB^{\dagger}A^\dagger)^{\frac{r}{2}}}\right)^{\frac{1}{r}} 
\\&= \left(\tr{\big( \sigma^{\frac{1-\alpha}{2}}\rho^{\frac{1}{2}}(\rho^{\frac{\alpha-1}{2}})^{2}  \rho^{\frac{1}{2}}\sigma^{\frac{1-\alpha}{2}}\big)^{\frac{2}{2}}}\right)^{\frac{1}{2}}
=\tr{\rho^{\alpha}\sigma^{1-\alpha}}^{\frac{1}{2}}.
\end{align}
Moreover, we have 
\begin{align}
\|A\|_p &= \left(\tr{(AA^{\dagger})^{\frac{p}{2}}}\right)^{\frac{1}{p}} = \left(\tr{ \big(\sigma^{\frac{1-\alpha}{2}}\rho^{\frac{1}{2}}\rho^{\frac{1}{2}} \sigma^{\frac{1-\alpha}{2}} \big)^{\frac{1}{2-\alpha}} }\right)^{\frac{2-\alpha}{2}} =\left(\tr{ \big(\sigma^{\frac{1-\alpha}{2}}\rho \sigma^{\frac{1-\alpha}{2}} \big)^{\frac{1}{2-\alpha}} }\right)^{\frac{2-\alpha}{2}},
\\ \|B\|_q &= \left(\tr{(BB^{\dagger})^{\frac{q}{2}}}\right)^{\frac{1}{q}} =\left(\tr{(\rho^{\frac{\alpha-1}{2}}\rho^{\frac{\alpha-1}{2}})^{\frac{1}{\alpha-1}}}\right)^{\frac{\alpha-1}{2}}= \left(\tr{\rho}\right)^{\frac{\alpha-1}{2}}=1.
\end{align}
Hence, we have
\begin{align}
\left(\tr{ \big(\sigma^{\frac{1-\alpha}{2}}\rho \sigma^{\frac{1-\alpha}{2}} \big)^{\frac{1}{2-\alpha}} }\right)^{\frac{2-\alpha}{2}}&= \|A\|_p \|B\|_q \ge   \|AB\|_{r} =\tr{\rho^{\alpha}\sigma^{1-\alpha}}^{\frac{1}{2}}.
\end{align}
Therefore by taking the logarithm in both sides then multiplying by $\frac{2}{\alpha-1}$ we obtain the desired inequality
\begin{align}
\widetilde{D}_{\frac{1}{2-\alpha}}( \rho\| \sigma) &= \frac{2-\alpha}{\alpha-1}\log \mathrm{Tr} \left(\sigma^{\frac{1-\alpha}{2}}\rho \sigma^{\frac{1-\alpha}{2}}\right)^{\frac{1}{2-\alpha}}\ge  \frac{1}{\alpha-1}\log \tr{\rho^{\alpha}\sigma^{1-\alpha}}=D_{\alpha}( \rho\| \sigma) . 
\end{align}
\end{proof}


\section{Non-signaling quantum error exponents}
\label{app:fully-quantum}

In this section, we prove Propositions \ref{prop-dual-quantum} \& \ref{prop-exp-ach-quantum}.

\begin{prop}[Restatement of Proposition \ref{prop-dual-quantum}]
Let $\cW_{A\rightarrow B}$ be a quantum channel with Choi matrix $(J_{\cW})_{RB}$. The one-shot meta-converse and activated non-signaling error probabilities satisfy
\begin{align}
\eps_{}^{\rm{NS}}(2M,\cW\otimes \cI_2) = \eps_{}^{\rm{MC}}(M,\cW) &= \inf_{\rho_R \in \cS(R)} \sup_{Z_B\mge 0} \tr{\sqrt{\rho_R} (J_{\cW})_{RB} \sqrt{\rho_R}\wedge M\rho_R\otimes Z_B}-\tr{Z_B}. 
\end{align}
\end{prop}

\begin{proof}
The proof is similar to \cite{matthews2012linear}. The meta-converse program can be written as \cite{Wang2019Feb}
\begin{equation}
\begin{split}
1-\eps_{}^{\rm{MC}}(M,\cW) =\sup_{\rho_R\in \cS(R)}\sup_{\Lambda_{RB}} &\quad\frac{1}{M} \tr{\Lambda_{RB} \cdot (J_{\cW})_{RB}}  
\\\text{s.t.}& \quad  \Lambda_B \mle  \mathbb{I}_B,
\\&\quad 0\mle  \Lambda_{RB} \mle M\rho_R \otimes  \mathbb{I}_B.
\end{split}
\end{equation}
Now, we consider the optimization over $\Lambda_{RB}$. The Lagrangian of the maximization program \eqref{NS-program} is
\begin{align}\label{lagr}
\cL(\Lambda_{RB}, \rho_R, \alpha ,Z_B, X_{RB}) &=  \frac{1}{M} \tr{\Lambda_{RB} \cdot (J_{\cW})_{RB}}  + \mathrm{Tr}\Big[ Z_B\Big(\dI_B- \Lambda_B\Big) \Big]+ \tr{X_{RB}(M\rho_R\otimes \dI_B-\Lambda_{RB}) }
\\&=  \mathrm{Tr}\Big[\Lambda_{RB} \Big(\frac{1}{M} (J_{\cW})_{RB} -X_{RB}-\dI_R\otimes Z_B\Big)\Big]+M\tr{\rho_RX_R}+\tr{Z_B},
\end{align}
where $Z_B\mge 0$ and $ X_{RB}\mge 0$.
Hence we should have $X_{RB}\mge \frac{1}{M}(J_{\cW})_{RB} -\dI_R\otimes Z_B $ so by the strong duality (the primal is bounded and strictly feasible) 
\begin{align}
1- \eps_{}^{\rm{MC}}(M,\cW) 
&=  \sup_{\rho_R\in \cS(R)}\inf_{\substack{Z_B\mge 0,X_{RB}\mge 0\\X_{RB}\mge \frac{1}{M}(J_{\cW})_{RB} -\dI_R\otimes Z_B }} M\tr{\rho_RX_R}+\tr{Z_B}
\\ &= \sup_{\rho_R\in \cS(R)}\inf_{\substack{Z_B\mge 0,Y_{RB}\mge 0\\Y_{RB}\mge \frac{1}{M}\sqrt{\rho}_R( J_{\cW})_{RB}\sqrt{\rho}_R -\rho_R\otimes Z_B }} M\tr{Y_R}+\tr{Z_B}  
\\ &=   \sup_{\rho_R\in \cS(R)}\inf_{Z_B\mge 0} \tr{(\sqrt{\rho}_R ( J_{\cW})_{RB}\sqrt{\rho}_R - M\rho_R\otimes Z_B)_+}+\tr{Z_B},
\end{align} 
where in the second equality  we used the change of variable  $ Y_{RB} = \sqrt{\rho}_R X_{RB} \sqrt{\rho}_R$ and in the last equality we used $\inf\{\tr{Y} : Y\mge 0, Y\mge \zeta\} = \tr{\zeta_+}$. Finally since $\tr{\sqrt{\rho}_R ( J_{\cW})_{RB}\sqrt{\rho}_R}=1$ and $\eps_{}^{\rm{NS}}(2M,\cW\otimes \cI_2) = \eps_{}^{\rm{MC}}(M,\cW)$ \cite{Wang2019Feb} we deduce the desired result
\begin{align}
\eps_{}^{\rm{NS}}(2M,\cW\otimes \cI_2) = \eps_{}^{\rm{MC}}(M,\cW) &= \inf_{\rho_R \in \cS(R)} \sup_{Z_B\mge 0} \tr{\sqrt{\rho_R} (J_{\cW})_{RB} \sqrt{\rho_R}\wedge M\rho_R\otimes Z_B}-\tr{Z_B}. 
\end{align}
\end{proof}

\begin{prop}[Restatement of Proposition \ref{prop-exp-ach-quantum}]
Let $\cW_{A\rightarrow B}$ be a quantum channel with Choi matrix $(J_{\cW})_{RB}$, $r\ge 0$, and $\alpha\in(0,1]$. Then, we have that
\begin{align}
E^{\rm{ANS}}(r)\ge  \sup_{\rho_R \in \cS(R)}\inf_{\sigma_B\in \cS(B)}    \frac{1-\alpha}{\alpha} (D_{\alpha}(\sqrt{\rho_R} (J_{\cW})_{RB} \sqrt{\rho_R}\| \rho_R\otimes  \sigma_B)-r). 
\end{align}
\end{prop}

\begin{proof}
Using Proposition \ref{prop-dual-quantum} we have that
\begin{align}
\eps_{}^{\rm{NS}}(M, \cW\otimes \cI_2)  &\le \eps_{}^{\rm{NS}}(2M, \cW\otimes \cI_2)= \inf_{\rho_R \in \cS(R)} \sup_{Z_B\mge 0} \tr{\sqrt{\rho_R} (J_{\cW})_{RB} \sqrt{\rho_R}\wedge M\rho_R\otimes Z_B}-\tr{Z_B}.
\end{align}
Writing $Z_B\mge 0$ as  $Z_B= s\sigma_B$, where $s=\tr{Z_B}$ and $\sigma_B\in \cS(B)$
 and using  the inequality \eqref{Chernoff} we obtain for any $\alpha \in (0,1]$
\begin{align}
\eps^{\rm{NS}}(M, \cW\otimes \cI_2) &\le    \inf_{\rho_R \in \cS(R)} \sup_{\sigma_B \in \cS(B),\; s\ge 0}  \tr{\sqrt{\rho_R} (J_{\cW})_{RB} \sqrt{\rho_R} \wedge M\rho_R\otimes s\sigma_B}-s
\\&   \le \inf_{\rho_R \in \cS(R)} \sup_{\sigma_B \in \cS(B),\; s\ge 0}  \tr{\left(\sqrt{\rho_R} (J_{\cW})_{RB} \sqrt{\rho_R}\right)^{\alpha}\left( M\rho_R\otimes s\sigma_B \right)^{1-\alpha}}-s
\\&=  \inf_{\rho_R \in \cS(R)} \sup_{\sigma_B \in \cS(B),\; s\ge 0}  s^{1-\alpha} M^{1-\alpha} \exp\left((\alpha-1)D_{\alpha}(\sqrt{\rho_R} (J_{\cW})_{RB} \sqrt{\rho_R}\| \rho_R\otimes \sigma_B  ) \right)-s.
\end{align}
Using a Young type inequality (Lemma \ref{lem-young}) in the form $\sup_{s\ge 0} s^{1-\alpha}c-s= \kappa_{\alpha  }c^{\frac{1}{\alpha}}$, where $\kappa_{\alpha} = {\alpha} (1-\alpha)^{\frac{1-\alpha}{\alpha}}\le 1$ we get
\begin{align}
\eps^{\rm{NS}}(M, \cW\otimes \cI_2) &\le   \inf_{\rho_R \in \cS(R)} \sup_{\sigma_B \in \cS(B),\; s\ge 0}  s^{1-\alpha} M^{1-\alpha} \exp\left((\alpha-1)D_{\alpha}(\sqrt{\rho_R} (J_{\cW})_{RB} \sqrt{\rho_R}\| \rho_R\otimes \sigma_B  ) \right)-s
\\&\le\inf_{\rho_R \in \cS(R)} \sup_{\sigma_B \in \cS(B)}   M^{\tfrac{1-\alpha}{\alpha}} \exp\left( \frac{\alpha-1}{\alpha} D_{\alpha}(\sqrt{\rho_R} (J_{\cW})_{RB} \sqrt{\rho_R}\| \rho_R\otimes \sigma_B  ) \right).
\end{align}
Taking $M = e^{nr}$, $n$ iid copies of the channel we deduce that
\begin{align}
\eps^{\rm{NS}}(e^{nr}, \cW^{\otimes n} \otimes \cI_2)
&\le \inf_{\rho_R^n \in \cS(R^n)} \sup_{\sigma_B^n \in \cS(B^{n})}  e^{\frac{nr(1-\alpha)}{\alpha}}\exp\left( \frac{\alpha-1}{\alpha} D_{\alpha}(\sqrt{\rho_R^n} (J_{\cW})_{RB}^{\otimes n} \sqrt{\rho_R^n}\| \rho_R^n\otimes \sigma_B^n  ) \right)
\\&\le  \inf_{\rho_R \in \cS(R)}   e^{\frac{nr(1-\alpha)}{\alpha}}\exp\left( \frac{\alpha-1}{\alpha} \inf_{\sigma_B^n \in \cS(B^{n})}D_{\alpha}((\sqrt{\rho_R} (J_{\cW})_{RB} \sqrt{\rho_R})^{\otimes n}\|\rho_R^{\otimes n}\otimes \sigma_B^n  ) \right)
\\&= \inf_{\rho_R\in \cS(R)}   e^{\frac{nr(1-\alpha)}{\alpha}}\exp\left( n\frac{\alpha-1}{\alpha} \inf_{\sigma_B \in \cS(B)}D_{\alpha}(\sqrt{\rho_R} (J_{\cW})_{RB} \sqrt{\rho_R}\|\rho_R\otimes \sigma_B )\right),
\end{align}
where we use the additivity of the R\'enyi mutual information \cite[Lemma 7]{Hayashi2016Oct} (Lemma \ref{app-lem-additivity}) in the last equality. Taking the logarithm on both sides we obtain the desired achievability bound
\begin{align}
\frac{1}{n}\log\eps^{\rm{NS}}(e^{nr}, \cW^{\otimes n}\otimes \cI_2) &\le   - \sup_{\rho_R \in \cS(R)}\inf_{\sigma_B\in \cS(B)}    \frac{1-\alpha}{\alpha} (D_{\alpha}(\sqrt{\rho_R} (J_{\cW})_{RB} \sqrt{\rho_R}\| \rho_R\otimes  \sigma_B)-r). 
\end{align}
\end{proof}

\section{Technical lemmas}\label{sec:Lemmas}

\begin{lem}[\cite{Cheng2023Nov}]\label{lem:concavity of min}
The function $(A,B)\mapsto \tr{A\wedge B}$ is  concave. 
\end{lem}

We include a proof for the reader's convenience.

\begin{proof}[Proof of Lemma \ref{lem:concavity of min}]
We have that 
\begin{align}
\tr{A\wedge B} = \inf_{0\mle O \mle \dI} \tr{AO} +\tr{B(\dI-O)}
\end{align}
and hence $(A,B)\mapsto \tr{A\wedge B}$ is an infimum of linear functions, so it is concave. 
\end{proof}

\begin{lem}\label{lem-young}
We have for $c > 0$ and $\alpha \in (0,1)$ that
\begin{align}
\sup_{s\ge 0} s^{1-\alpha}c-s= \frac{\alpha}{1-\alpha} (1-\alpha)^{\frac{1}{\alpha}}c^{\frac{1}{\alpha}}.
\end{align}
\end{lem}

\begin{proof}
This result follows from Young inequality, which we recall.  For all $x,y\ge 0$ and $p,q > 1$ satisfying $\frac{1}{p}+\frac{1}{q}=1$, we have that
\begin{equation}
\frac{x^p}{p} + \frac{y^q}{q} \ge xy.
\end{equation}
We apply this inequality with $ p =\frac{1}{\alpha}$, $q= \frac{1}{1-\alpha}$, $x=(1-\alpha) c $ and $y= s^{1-\alpha} $ to obtain 
\begin{equation}
\alpha(1-\alpha)^{\frac{1}{\alpha}}c^{\frac{1}{\alpha}}+(1-\alpha)s \ge (1-\alpha)cs^{1-\alpha}
\end{equation}
with equality for $s = (1-\alpha)^{\frac{1}{\alpha}}c^{\frac{1}{\alpha}}$.
\end{proof}

\begin{lem}\label{lem-perm-inv}
Let $X$ be a compact convex set. Let $G$ be a finite group acting on $X$. Let $f$ be a convex real valued function on $X$. If $f$ is invariant under the action of $G$, i.e.,
\begin{align}
\forall x\in X, \forall g\in G : f(g\cdot x)=f(x),
\end{align}
then $f$ achieves its infimum at a point  invariant under the action of $G$. 
\end{lem}

\begin{proof}
Let $x^\star$ such that $\inf_x f(x)= f(x^\star)$. By the convexity of $f$ we have that:
\begin{align}
f(x^\star)\le f\left(\frac{1}{|G|}\sum_{g\in G}g\cdot x^\star\right)\le \frac{1}{|G|}\sum_{g\in G}f\left(g\cdot x^\star\right)=f(x^\star).
\end{align}
Hence, we have
\begin{align}
\inf_x f(x)= f(x^\star)= f\left(\frac{1}{|G|}\sum_{g\in G}g\cdot x^\star\right)
\end{align}
and $\frac{1}{|G|}\sum_{g\in G}g\cdot x^\star$ is invariant under the action of $G$.  
\end{proof}

\begin{lem}[\cite{Hayashi2016Oct}, Lemma 7]\label{app-lem-additivity}
Define the mutual information of order $\alpha$ for $\rho_{AB} \in \cS(AB)$ and $\tau_A\mge 0$ such that $\tau_A\gg \rho_A$ as
\begin{align}
I_\alpha(\rho_{AB}\| \tau_A ) &:= \inf_{\sigma_B \in \cS(B)} D_\alpha(\rho_{AB}\| \tau_A \otimes \sigma_B ).
\end{align}
Let $\rho_{AB} \in \cS(AB)$, $\omega_{A'B'} \in \cS(A'B')$, and $\tau_A, \sigma_A\mge 0$ with $\tau_A\gg \rho_A$, $\sigma_{A'}\gg \omega_{A'}$. Then, we have that
\begin{align}
I_\alpha(\rho_{AB}\otimes \omega_{A'B'} \| \tau_A \otimes \sigma_{A'} )&=I_\alpha(\rho_{AB}\| \tau_A )+ I_\alpha(\omega_{A'B'}\| \sigma_{A'} ), \quad \forall \alpha\ge 0.
\end{align}
\end{lem}


\end{document}